\documentclass[11pt]{amsart}

\usepackage{amssymb}
\usepackage[colorlinks,citecolor=blue,urlcolor=blue,bookmarks=false,hypertexnames=true]{hyperref}

\usepackage[a4paper,margin=23mm,top=30mm]{geometry}

\usepackage{amsfonts, amsmath, amssymb, amsgen, amsthm, amscd}
\usepackage{newtxtext,newtxmath}
\usepackage[utf8]{inputenc} 

\usepackage{color}

\usepackage[all]{xy}

\usepackage[colorlinks]{hyperref}
\usepackage{mathtools,slashed}
\usepackage{setspace}
\usepackage{enumerate}

\def\d{\delta}

\def\om{\omega}

\def\s{\sigma}

\def\t{\theta}

\def\vp{\varphi}

\def\Id{\mathop{\rm Id}\nolimits}
\def\p{\partial}

\def\ot{\otimes}
\def\odots{\ot\cdots\ot}

\def\rt{\triangleright}
\def\lt{\triangleleft}

\def\Id{\mathop{\rm Id}\nolimits}

\def\ad{\mathop{\rm ad}\nolimits}

\usepackage{enumerate}

\newcommand{\G}[1]{\mathfrak{#1}}
\newcommand{\C}[1]{\mathcal{#1}}
\newcommand{\B}[1]{\mathbb{#1}}

\renewcommand{\leq}{\leqslant}
\renewcommand{\geq}{\geqslant}

\numberwithin{equation}{section}

\newtheorem{theorem}{Theorem}[section]
\newtheorem{proposition}[theorem]{Proposition}

\newtheorem{corollary}[theorem]{Corollary}
\theoremstyle{definition}

\newtheorem{remark}[theorem]{Remark}

\author{O\u{g}ul Esen}
\address{Department of Mathematics, Gebze Technical University,  41400 Gebze-Kocaeli, Turkey}
\email{oesen@gtu.edu.tr}

\author{Serkan Sütlü}
\address{Department of Mathematics, I\c{s}ik University, 34980 \c{S}ile-\.{I}stanbul, Turkey}

\email{serkan.sutlu@isikun.edu.tr}

\begin{document}
\title{Matched pair analysis of the Vlasov plasma}
\date{}
\begin{abstract}
We present the Hamiltonian (Lie-Poisson) analysis of the Vlasov plasma, and the dynamics of its kinetic moments, from the matched pair decomposition point of view. We express these (Lie-Poisson) systems as couplings of \textit{mutually interacting} (Lie-Poisson) subdynamics. The mutual interaction is beyond the well-known semi-direct product theory. Accordingly, as the geometric framework of the present discussion, we address the \textit{matched pair Lie-Poisson} formulation allowing mutual interactions. Moreover, both for the kinetic moments and the Vlasov plasma cases, we observe that one of the constitutive subdynamics is the compressible isentropic fluid flow, and the other is the dynamics of the kinetic moments of order $\geq 2$. In this regard, the algebraic/geometric (matched pair) decomposition that we offer, is in perfect harmony with the physical intuition. To complete the discussion, we present a momentum formulation of the Vlasov plasma, along with its matched pair decomposition.   
\end{abstract}
\maketitle
\setcounter{tocdepth}{2}
\tableofcontents

\onehalfspace
\setlength{\parskip}{0.5cm}

\section{Introduction} 
Once a Hamiltonian realization of a physical system has been achieved, the analysis of many qualitative aspects of the system; such as the control, integrability, stability, and the asymptotic behaviour, become much more accessible  \cite{abraham1978foundations,arnol2013mathematical}. As such, in recent years, many physical systems have been studied in the realm of the Hamiltonian dynamics; from the classical models to continuum, as well as the  field theories. We refer the reader to \cite{binz2011geometry,LeRo,holm2009geometric,
libermann2012symplectic,MarsdenRatiu-book} for an incomplete list of examples. 

The present paper is on the (de)coupling problem of the Hamiltonian systems, which may be summarized along the following lines. Given two dynamical systems in interaction, the equation of motion of the coupled system does not merely consists of the equations of motions of the individual systems. Instead, there additional terms appear as a reflection of the interaction between the constitutive systems. 

The examples, built on a one-way action (only one of the systems effects the other), fell in the realm of the semi-direct product\footnote[1]{We shall, following the terminology of \cite{Majid-book} - upon which the algebraic framework of the present paper is built, prefer ``semi-direct sum'' Lie algebras. ``Semi-direct product'', on the other hand, will be reserved for Lie groups.} theory,  \cite{holm1982noncanonical,
marsden1984reduction,MarsRatiWein84}. The coupling of the Maxwellian property of a continuum with its fluid motion, such as the magnetohydrodynamics \cite{holm1982noncanonical} or the Maxwell-Vlasov system \cite{MaWe81}, are the quintessential examples studied by the semi-direct product theory.

On the other hand, the Hamiltonian systems built on a pair of mutually interacting subsystems, has been studied only recently in the Lie-Poisson formulation \cite{EsSu16}, and has already found applications in the field of reversible thermodynamics and kinetic models \cite{EsGrGuPa19,esen2017hamiltonian}. The algebraic foundations of this more general theory lie in the \textit{matched pair} construction of \cite{Maji90}. The generality is reflected by the fact that in case the effect of one of the subsystems is assumed to be trivial, then the matched pair theory reduces to the semi-direct product theory. Let us note that although we follow the terminology of matched pairs from \cite{Maji90,Maji90-II,Majid-book}, the construction appears in the literature in different names (for different purposes); such as the twilled extension in \cite{KoMa88}, the double Lie group in \cite{LuWein90}, or the Zappa-Sz\'ep product in \cite{Br05}. 

Now, the main goal of the present paper may be stated as to show that the equations governing the Vlasov plasma and the kinetic moments both admit matched pair decompositions. More precisely, their dynamics may be built upon the dynamics of the isentropic compressible fluid motion, and the dynamics of the kinetic moments of order $\geq 2$.

It worths mentioning that the matched pair approach appears to be a fertile strategy. In \cite{EsenSutl17}, it is used in the study of  the Lagrangian dynamics on Lie groups, wherein the Lagrangian and the Hamiltonian dynamics of matched pairs are linked by proper Legendre transformations. The matched pair strategy was also applied successfully to the higher order Lagrangian systems in \cite{esen2019matched}, and to the discrete dynamics in the Lie groupoid setting in \cite{esen2018matched}. Let us next present a brief review of the literature; in order to be able to state what exactly is missing in the literature, and how the present paper aims to fills it. 

\subsubsection*{\textbf{Lie-Poisson equations}} 

Let us begin with the Lie-Poisson formulation \cite{holm2008geometric,MarsdenRatiu-book} which provides a tangible framework that many continuum models fit. Let a continuum rests in a finite region $\C{Q}\subset \mathbb{R}^3$ without boundary. Then, the symmetry group of a physical motion is an infinite dimensional Lie group, say $K$, preserving the motion. In particular, the symmetry group of the incompressible isentropic Euler's fluid is the group of volume preserving diffeomorphisms on $\C{Q}$, while the symmetry group of the Vlasov plasma is the group of canonical diffeomorphisms on the cotangent bundle $T^*\C{Q}$. The (Lie-Poisson) equation that govern the motion, on the other hand, is given on the linear algebraic dual $\G{K}^*$ of the Lie algebra $\G{K}$. The dual space $\G{K}^*$ is  a Poisson space, called the Lie-Poisson space, on which the Lie-Poisson bracket is given by
\begin{equation} \label{LP-bracket}
\{\C{H},\C{G}\}(z)=-\Big \langle z,\big[\frac{\delta \C{H}}{\delta z},\frac{\delta \C{G}}{\delta z} \big] \Big \rangle,
\end{equation}
for any $z \in \G{K}^*$, and any $\C{H},\C{G}\in\G{K}^*$. Let us note that ${\delta \C{H}}/{\delta z}$ and ${\delta \C{G}}/{\delta z}$ denote the Fr\'{e}chet derivatives of the functionals, and (assuming the reflexivity) they are elements of $\G{K}$. Accordingly, the equation of motion of the system, due to a Hamiltonian functional $\C{H}$, is computed to be
\begin{equation} \label{LP-eqn}
\frac{\partial z}{\partial t}=-\ad^*_{{\delta \C{H}}/{\delta z}} z,
\end{equation}
in terms of the (left) coadjoint action, which corresponds to the negative of the linear algebraic dual of the (left) adjoint action. That is,
\begin{equation} \label{coad}
\langle \ad^\ast_{x} z, x' \rangle =-\langle z, \ad_{x}x'\rangle = - \langle \mu, [x,x']\rangle.
\end{equation}

\subsubsection*{\textbf{Hamiltonian analysis of the Vlasov plasma}} 
Let us now consider the non-relativistic collisionless plasma particles in $\C{Q}$, and the momentum-phase space $T^\ast \C{Q}$, with the Darboux' coordinates $(q^i,p_j)$. The plasma dynamics is determined by the evolution of the plasma density function $f$, defined on $T^\ast \C{Q}$, according to the Vlasov equation
\begin{equation}\label{Boltzmann-intro}
\frac{\partial f}{\partial t}+\frac{1}{m}\delta^{ij} p_i \frac{\partial f}{\partial q^j} -e \frac{\partial \phi}{\partial q^i}   \frac{\partial f}{\partial p_i}=0,
\end{equation}
where $\phi$ is the potential function, $m$ is the mass, and $e$ is the electrical charge. Coupling the Vlasov equation with the Poisson equation
\begin{equation}\label{poi}
\nabla^{2}\phi(q)=-e\int f(q,p)dp,
\end{equation}
in the case of the non-relativistic framework, one arrives at the Vlasov-Poisson equations. Coupling the Vlasov equation \eqref{Boltzmann-intro} with the Maxwell equations, on the other hand, one obtains the Maxwell-Vlasov equations. 

Hamiltonian analysis of the Vlasov plasma was achieved in \cite{MaWe81} through the symmetry group ${\rm Diff}_{\rm can}(T^*\C{Q})$ of canonical diffeomorphisms on $T^\ast \C{Q}$, which acts on $T^*\C{Q}$ from the right \cite{marsden1983hamiltonian}, known as the particle relabelling symmetry. 
The Lie algebra of ${\rm Diff}_{\rm can}(T^*\C{Q})$ is identified with the space $\C{F}(T^\ast\C{Q})/\B{R}$ of smooth functions modulo the constants, equipped with the opposite (\textit{minus}) canonical Poisson bracket, as a manifestation of the right symmetry. Let us note also that once the symplectic volume $dqdp$ is fixed as the top form, the $L_2$-pairing between $\C{F}(T^*\C{Q})$ and its dual $\C{F}^\ast(T^*\C{Q})$ - which happens to be the space of densities $Den(T^*\C{Q})$ on $T^*\C{Q}$ - allows to identify the dual space $\C{F}^\ast(T^*\C{Q})$ with $\C{F}(T^*\C{Q})$ itself. Accordingly, the Lie-Poisson bracket \eqref{LP-bracket}, for two functionals $\C{H}$ and $\C{G}$ on $Den(T^*\C{Q})$, takes the particular form
\begin{equation} \label{Vlasov-bracket}
\{\C{H},\C{G}\}^V(f)=\int_{T^*\C{Q}}\, f~\big\{\frac{\delta \C{H}}{\delta f}, \frac{\delta \C{G}}{\delta f}\big\}\, dqdp.
\end{equation}
The bracket that appear in the integral is the canonical Poisson bracket. Then, the Vlasov equation \eqref{Boltzmann-intro} may be written as a Lie-Poisson equation in this dual space, if the Hamiltonian functional is assumed to be $\C{H}(h)=\int (hf) dqdp$. Here, $h=p^2/2m+e\phi$ is taken to be total energy of a single particle motion, in which case the Fr\'{e}chet derivative $\delta \C{H}/ \delta f$ is equal to $h$. A straightforward calculation reveals that the coadjoint action on $Den(T^*\C{Q})$ is given by the opposite canonical Poisson bracket, and the Vlasov equation \eqref{Boltzmann-intro} may be written in the form
\begin{equation}\label{Ham-Vlasov}
\frac{\partial f}{\partial t}=-\ad^*_{\frac{\partial \C{H}}{\partial f}}f=
\{h,f\}.
\end{equation}
We refer the reader to \cite{morrison1982poisson,PaKiEsGr16} for further details on the Hamiltonian realization of the Vlasov motion. 

\subsubsection*{\textbf{Kinetic moments of the Vlasov plasma}}

As discussed in the previous paragraph, the motion of the plasma is governed by the evolution of the density function $f$ in the Vlasov equation \eqref{Boltzmann-intro}. The kinetic moments \cite{chapman1939mathematical} of the density function is determined through the integral hierarchy
\begin{equation} \label{kinetic moments-intro}
\mathbb{A}_{i_1 \dots i_m}(q)=\int_{T^*_q\C{Q}}\,p_{i_1}\dots p_{i_m} f(q,p)\,dp
\end{equation}
for all non-negative integers $m\geq 0$. As such, the kinetic moments determine the symmetric covariant tensor fields. Accordingly, the dynamics of the kinetic moments may be expressed as a Lie-Poisson structure as follows. The space ${\mathfrak{T}\mathcal{Q}}$ of symmetric contravariant tensor fields on $\C{Q}$  carries a (graded) Lie algebra structure via the symmetric Schouten concomitant \cite{KoMiSl93,Ma97,Sc40,Tr08}. For a $k$-th order contravariant field $\mathbb{X}^k$ and an $m$-th order contravariant field $\mathbb{Y}^m$, the symmetric Schouten concomitant is defined to be the contravariant tensor field 
\begin{equation} \label{SC-def-intro}
\left[ \mathbb{X}^k,\mathbb{Y}^m\right]:=\left(k\mathbb{X}^{i_{m+1}...i_{m+k-1}\ell}\mathbb{Y}^{i_{1}...i_{m}}_{,\ell}
-
m\mathbb{Y}^{i_{k+1}...i_{k+m-1}\ell}  \mathbb{X}^{i_{1}i_{2}...i_{k}}_{,\ell}\right)\, \partial {q^{i_{1}}}\otimes ...\otimes \partial {q^{i_{k+m-1}}}.
\end{equation}
of order $m+k-1$. We refer the reader to Subsection \ref{subsect-symm-contra-tensor} for further details. 

On the other hand, once a volume form $dq$ on $\C{Q}$ is fixed, we can consider the space ${\mathfrak{T}^*\mathcal{Q}}$ of symmetric covariant tensor fields as the dual of the symmetric contravariant tensor fields ${\mathfrak{T}\mathcal{Q}}$. Let us note that the kinetic moments \eqref{kinetic moments-intro} are elements of ${\mathfrak{T}^*\mathcal{Q}}$. Being the dual of a Lie algebra, ${\mathfrak{T}^*\mathcal{Q}}$ carries a Lie-Poisson bracket called, in literature, the Kupershmidt–Manin bracket. The Kupershmidt-Manin bracket, for two functionals $\C{H}$ and $\C{G}$ on ${\mathfrak{T}^*\mathcal{Q}}$, is
\begin{equation} \label{KM-bracket}
\{\C{H},\C{G}\}^{KM} (\mathbb{A}_{m+k-1})=-\int \mathbb{A}_{m+k-1}\cdot \big[ 
\frac{\delta \C{H}}{\delta \mathbb{A}_m},
\frac{\delta \C{G}}{\delta \mathbb{A}_k}
\big] dq,
\end{equation} 
where the bracket inside the integral is the symmetric Schouten concomitant \eqref{SC-def-intro}. In \cite{gibbons1981collisionless}, it is established that the kinetic moments are actually Poisson mappings from $Den(T^*\C{Q})$ to ${\mathfrak{T}^*\mathcal{Q}}$, respecting the Vlasov bracket \eqref{Vlasov-bracket} and the Kupershmidt-Manin bracket \eqref{KM-bracket}. We refer the reader to the recent papers \cite{GiHoTr08,gibbons2008vlasov} for such an analysis of the kinetic moments. These papers have also  motivational importance for the present study. 

\subsubsection*{\textbf{Plasma-to-fluid map}}

The semi-direct product $\G{s}=\C{F}(\C{Q})  \rtimes \mathfrak{X}(\mathcal{Q})$ of the zeroth order tensor fields $\C{F}(\C{Q})$, that is the space of smooth functions, with the first order contravariant tensor fields $\mathfrak{X}(\mathcal{Q})$, namely the space of smooth vector fields, constitutes a Lie subalgebra of the symmetric contravariant tensor fields ${\mathfrak{T}\mathcal{Q}}$ of all orders. Let us note that $\G{s}$ is the largest Lie subalgebra containing $\C{F}(\C{Q})$. 
Now, given two elements $\hat{\lambda} :=(\eta,Z) $ and $\hat{\s}:=(\sigma,Y )$ in $\C{F}(\C{Q})  \rtimes \mathfrak{X}(\mathcal{Q})$, the semi-direct product bracket is computed from \eqref{SC-def-intro} as
\begin{equation}\label{pa-intro}
 [  (\eta,Z  ) , (\sigma,Y  )  ]
= (Z(\sigma) -Y(\eta), [ Z,Y]  ),
\end{equation}
where $Z(\sigma)$ and $Y(\eta)$ are the directional derivatives, whereas the latter bracket is the Jacobi-Lie bracket of vector fields. The linear dual of $\G{s}$, on the other hand, may be given by $\G{s}^*=\C{F}(\C{Q})  \oplus \Lambda^1(\mathcal{Q})$, where $\Lambda^1(\mathcal{Q})$ denotes the space of 1-forms on $\C{Q}$. Let us note also that the dual space $\G{s}^\ast$ is the configuration space of the 
compressible isentropic fluid flow \cite{marsden1984reduction,MarsRatiWein84,
marsden1983coadjoint}, where the compressible fluid  bracket is given by 
\begin{equation}\label{CF}
\{ \C{H},\C{G}\}^{CF}( \rho ,M)=-\int_{\mathcal{Q}}\left\langle M,\left[\frac{\delta \C{H}}{\delta M},
\frac{\delta \C{G}}{\delta M} \right]\right\rangle  -\rho \left( {\frac{\delta\C{H}}{
\delta M}}\left( \frac{\delta \C{G}}{\delta \rho }\right) -{\frac{\delta \C{G}}{\delta M}}\left( \frac{\delta \C{H}}{
\delta \rho }\right) \right) dq.
\end{equation}
A direct  observation reveals that the first two kinetic moments \eqref{kinetic moments-intro} of the plasma density function, that is,
\begin{equation}
\rho(q)=\mathbb{A}_{0}(q)=\int_{T^*_q\C{Q}}f(q,p)dp, \qquad M_i=\mathbb{A}_{i}(q)=\int_{T^*_q\C{Q}}p_if(q,p)dp
\end{equation}
determine a Poisson map from $Den(T^*\C{Q})$ to $\G{s}^\ast$, compatible both with the Vlasov bracket \eqref{Vlasov-bracket} and the compressible fluid bracket \eqref{CF}.

\subsubsection*{\textbf{Hamiltonian dynamics on matched pairs}}

The algebraic foudations of our discussions lie in the theory of matched pairs of Lie algebras, and their Lie-Poisson counterparts, presented in Section \ref{Sec-DMIS}. Given two Lie algebras $\G{g}$ and $\G{h}$ with mutual interactions (given by compatible Lie algebra actions), their direct sum $\G{g}\oplus \G{h}$ may be endowed with a Lie algebra structure. In this case, the Lie algebra $\G{g}\oplus \G{h}$ is denoted by $\G{g}\bowtie \G{h}$, and is called the double cross sum Lie algebra \cite{Maji90-II,Maji90,Majid-book,Ta81}.  Moreover, the Lie bracket on the matched pair Lie algebra is given by 
\begin{equation}\label{mpla-intro}
\big[ (\xi\oplus \eta),(\xi'\oplus \eta')\big ]_{\bowtie}=\underbrace{\big( [\xi,\xi']+\eta\rt \xi'-\eta'\rt \xi \big)}_{\in ~\G{g}}
\oplus \underbrace{\big( [\eta,\eta']+\eta\lt \xi'-\eta' \lt \xi \big)}_{\in ~\G{h}},
\end{equation}
for any $\xi\oplus\eta$ and $\xi'\oplus\eta'$ in $\G{g}\bowtie \G{h}$, where the bracket in the first summand is the Lie bracket on $\G{g}$, while the one in the latter summand is the Lie bracket on  $\G{h}$. The remaining terms may be considered as the incarnations of the mutual interactions of $\G{g}$ and $\G{h}$. It worths mentioning that if one of the actions is trivial then a double cross sum Lie algebra reduces to a semi-direct sum Lie algebra. That is, the matched pair construction is a generalization of the semi-direct sum construction.  

From the decomposition point of view, if a Lie algebra $\G{K}$ can be decomposed (as a vector space) into a direct sum of two Lie subalgebras, say $\G{K}\cong\G{g}\oplus \G{h}$, then $\G{K}\cong\G{g}\bowtie \G{h}$ as Lie algebras. In this case, the mutual actions may be obtained from 
\begin{equation}\label{mpla-intro-2}
[\eta,\xi' ]=\underbrace{(\eta\rt \xi')}_{\in ~ \G{g}}
\oplus  \underbrace{( \eta\lt \xi')}_{\in ~ \G{h}}.  
\end{equation}
for any $\eta$ in $\G{h}$, and any $\xi'$ in $\G{g}$.

The matched pair decomposition $\G{K}=\G{g}\bowtie \G{h}$ leads to a matched pair decomposition of the Lie-Poisson bracket on $\G{K}^\ast = \G{g}^*\oplus \G{h}^*$, and accordingly, a matched pair decomposition of the Lie-Poisson equations. 

Let us recall the matched pair realization of the Lie-Poisson bracket from \cite{EsSu16}. Given  two function(al)s, $\C{H}=\C{H}(\mu,\nu)$ and $\C{G}=\C{G}(\mu,\nu)$ on the dual space $\G{K}^*$, the Lie-Poisson bracket turns out to be

\begin{equation} 
\begin{split} \label{LiePoissonongh-intro}
 \left\{ \mathcal{H},\mathcal{G}\right\}(\mu, \nu) &=\underbrace{-\left\langle \mu ,\left[\frac{\delta \mathcal{H}}{%
\delta \mu},\frac{\delta \mathcal{G}}{\delta \mu}\right]\right\rangle
-\left\langle \nu ,\left[\frac{\delta \mathcal{H}}{\delta \nu},\frac{\delta
\mathcal{G}}{\delta \nu}\right]\right\rangle}_{\text{A: direct product}} 
-\underbrace{\left\langle \mu ,\frac{\delta \mathcal{H}}{\delta \nu}
\rt \frac{\delta \mathcal{G}}{\delta \mu}\right\rangle
+\left\langle \mu ,\frac{\delta \mathcal{G}}{\delta \nu}\rt
\frac{\delta \mathcal{H}}{\delta \mu}\right\rangle}_ {\text{B: via the left action of $\G{h}$ on 
$\G{g}$}} \\
&  \qquad -\underbrace{ \left\langle \nu ,\frac{\delta \mathcal{H}}{\delta \nu}\lt \frac{\delta \mathcal{G}}{\delta \mu}\right\rangle
+\left\langle \nu ,\frac{\delta \mathcal{G}}{\delta \nu}\lt
\frac{\delta \mathcal{H}}{\delta \mu}\right\rangle} _{\text{C: via the
right action of $\G{g}$ on 
$\G{h}$}}.
\end{split}
\end{equation}
We note that, the label A refers to a sum of individual Poisson brackets on the dual spaces $\G{g}^* $ and $\G{h}^* $. The label B, on the other hand, follow from the left action of $\G{h}$ on $\G{g}$, while C is the result of the right action of $\G{g}$ on $\G{h}$. In case of a one-sided action, as in semi-direct sum theories, B or C drops. If, furthermore, there is no action, then both  B and C vanish. 

Accordingly, the Lie-Poisson equation \eqref{LP-eqn}, generated by a Hamiltonian function 
$\mathcal{H}=\mathcal{H}(\mu,\nu)$ on $\G{K}^\ast=\G{g}^\ast\oplus\G{h}^\ast$ is decomposed into two equations of the form
\begin{align}\label{LPEgh-intro}
\begin{split}
& \underbrace{\frac{d\mu}{dt} = -
\ad^{\ast}_{\frac{\delta\mathcal{H}}{\delta\mu}}(\mu)}_{\text{Lie-Poisson Eq. on }\G{g}^*}
+
\underbrace{\mu\overset{\ast }{\lt}
\frac{\delta\mathcal{H}}{\delta\nu}}
_{\text{action of } \G{h}}
+
\underbrace{\mathfrak{a}_{\frac{\delta\mathcal{H}}{\delta\nu}}^{\ast}\nu}
_{\text{action of }\G{g}}, 
\\
&\underbrace{\frac{d\nu}{dt} =
-
\ad^{\ast}_{\frac{\delta\mathcal{H}}{\delta\nu}}(\nu)}_
{\text{Lie-Poisson Eq. on }\G{h}^*}
-
\underbrace{
\frac{\delta\mathcal{H}}{\delta\mu} \overset{\ast }{\rt}\nu}_{\text{action of }\G{g}}
- \underbrace{\mathfrak{b}
_{\frac{\delta\mathcal{H}}{\delta\mu}}^{\ast}\mu
}_{\text{action of }\G{h}}
.
\end{split}
\end{align}
The first terms on the right hand sides are the individual Lie-Poisson equations on $\G{g}^\ast$ and $\G{h}^\ast$, respectively. The other terms are  dual and cross actions appear
 as the manifestations of the mutual Lie algebra actions. In Subsection \ref{Subsec-mp}, we shall precisely define the dual mappings $\overset{\ast }{\lt}$ and $\overset{\ast }{\rt}$, as well as the cross actions $\mathfrak{a}^*$ and $\mathfrak{b}^*$. We are now ready to state our first goal.

\subsubsection*{\textbf{The first goal of the present work. Matched pair decomposition of the kinetic moments.}} So far, we have mentioned the Lie algebra ${\mathfrak{T}\mathcal{Q}}$  of symmetric contravariant tensor fields of all orders equipped with the symmetric Schouten concomitant \eqref{SC-def-intro}, and its Lie subalgebra $\G{s}=\C{F}(\C{Q})  \rtimes \mathfrak{X}(\mathcal{Q})$ with the semi-direct sum Lie bracket \eqref{pa-intro}. Let us note also the complement $\G{n}$ of $\G{s}$ in ${\mathfrak{T}\mathcal{Q}}$, that is ${\mathfrak{T}\mathcal{Q}}=\G{s}\oplus \G{n}$, which consist of the symmetric contravariant tensors of order $\geq 2$. It follows at once from the graded character of the Schouten concomitant \eqref{SC-def-intro} that  $\G{n}$ is also a Lie subalgebra of $\mathfrak{T}\mathcal{Q}$. Hence, it follows from the universal property of the double cross sum construction that ${\mathfrak{T}\mathcal{Q}} =\G{s}\bowtie \G{n}$, which is our first objective. In view of \eqref{mpla-intro-2}, then, it suffices to compute $[\G{n},\G{s}]$ in ${\mathfrak{T}\mathcal{Q}}$ for the mutual actions. Being the dual of a matched pair Lie algebra, the Kuperschmidt-Manin bracket \eqref{KM-bracket} thus admits a matched pair decomposition in the form of \eqref{LiePoissonongh-intro}. Accordingly, the dynamics of the kinetic moments may, in turn, be given by the Lie-Poisson equations in the form \eqref{LPEgh-intro}. We refer the reader to Subsection \ref{subsect-kinetic-moments-dynamics} for details of the last two assertions. Here, the dynamics on $\G{s}^*$ encodes to that of the compressible fluid flow, whereas the dynamics on $\G{n}^*$ governs the motion of the kinetic moments of order $\geq 2$.

\subsubsection*{\textbf{The Vlasov plasma in momentum formulation}}

While investigating possible Hamiltonian realizations of the Vlasov plasma, in \cite{Gu10}, an intermediate level of equations has been introduced. In this intermediate level, which is connected to the classical Hamiltonian formalism of the Vlasov dynamics by means of a Poisson mapping, the Lie algebra of the configuration group ${\rm Diff}_{\rm can}(T^*\C{Q})$ appears to be the space $\mathfrak{X}_{\rm{Ham}}(T^\ast \C{Q})$ of Hamiltonian vector fields, rather then the function space. 

The dual space $\mathfrak{X}_{\mathrm{\mathrm{Ham}}}^*(T^\ast \C{Q})$, on the other hand, may be identified (see Proposition \ref{dualHam} below) with the space of 1-forms with non-trivial divergences, once the symplectic volume $dqdp$ is fixed as the top-form. We thus have the Vlasov bracket in momentum formulation, namely the momentum-Vlasov  bracket
\begin{equation}
\{\C{H},\C{G} \}^{mV}(\Pi_f)= \int_{T^*Q} \Pi_f\cdot \big [\frac {\partial \C{H}}{\partial \Pi_f},\frac {\partial \C{G}}{\partial \Pi_f}\big]\,dqdp,
\end{equation}
where the bracket on the right hand side is the Jacobi-Lie bracket of vector fields.
The coadjoint action of a Hamiltonian vector field $-X_h$ in $\mathfrak{X}_{\mathrm{\mathrm{Ham}}}(T^\ast \C{Q})$ on a 1-form $\Pi_f$ in $\mathfrak{X}_{\mathrm{\mathrm{Ham}}}^*(T^\ast \C{Q})$ is given by the opposite (minus) Lie derivative. 
In this representation, the Hamiltonian formulation turns out to be  a coadjoint flow described by 
\begin{equation} \label{mom-Vlasov-intro}
\frac{d\Pi_f}{dt}=-\ad^*_{X_h}\Pi_f=- \C{L}_{X_h}\Pi_f.
\end{equation}
Once again, the minus sign in front of the Hamiltonian vector field is a manifestation of the right symmetry. 
The system \eqref{mom-Vlasov-intro} is called the momentum-Vlasov equations \cite{esen2012geometry,Gu10}, and there exists a Poisson mapping (as the dual of a Lie algebra isomorphism) between the Vlasov equation \eqref{Ham-Vlasov} and the momentum-Vlasov equation \eqref{mom-Vlasov-intro}. 

\subsubsection*{\textbf{Advantages of the momentum-Vlasov approach.}} The momentum-Vlasov formulation \eqref{mom-Vlasov-intro} of the plasma dynamics has the following advantages. 

\textit{\textbf{1. The inverse Legendre transformation.}} In the  Hamiltonian formulation of the (Poisson-)Vlasov equation, the Hamiltonian functional is degenerate \cite{MaWe81,morrison1981hamiltonian}. As such, there is no direct way to define the (inverse) Legendre transformation in order to establish an equivalence between the Lie-Poisson and the Euler-Poincar\'e  formulations of the Vlasov dynamics. 

Furthermore, one has to overcome the complications related to the function spaces. The latter  involves the problem of the identification of the function space and its dual, since they are isomorphic under the $L_2$-pairing. The momentum-Vlasov equation \eqref{mom-Vlasov-intro} has been introduced, see \cite{EsGrGuPa19,esen2011lifts,esen2012geometry,Gu10}, to separate the Lie algebra and the dual space by considering the Lie algebra as the Lie algebra of Hamiltonian vector fields $\mathfrak{X}_{\mathrm{\mathrm{Ham}}}(T^\ast \C{Q})$, and the dual space as the 1-form densities  $\mathfrak{X}_{\mathrm{\mathrm{Ham}}}^*(T^\ast \C{Q})$. Proper definitions of the Lie algebra and its dual are especially important for the formulation of the (inverse) Legendre transformation of the Hamiltonian formulation of the (Poisson-)Vlasov plasma, which is still an open question. We shall, however, postpone the Legendre transformation to a future work, where we shall apply the Tulczyjew triplet \cite{tulczyjew1977legendre} to the Vlasov motion, which has been successfully exercised on the dynamics involving symmetries by Lie groups \cite{esen2014tulczyjew,esen2017tulczyjew}. We do hope that, the (matched pair) decomposition of the Hamiltonian formulation of both Vlasov and the momentum-Vlasov equations (which will be established in the present work) will lead to a more precise realization of the Tulczyjew triplet in this framework. Furthermore, the matched pairs seem to provide a promising formalism to define a proper (inverse) Legendre transformation to the Hamiltonian dynamics.  

\textit{\textbf{2. A geometric pathway to plasma dynamics.}} Another important consequence of the momentum realization of the plasma motion is the existence of a \textit{geometric pathway} to the momentum-Vlasov equation \eqref{mom-Vlasov-intro}, starting from the single particle motion. This geometry that the momentum-Vlasov equation admit is indeed proper to the kinetic theory \cite{VeSiDu11}. We find this harmony of the physical intuition and the geometrical realization interesting, and worth mentioning. To describe the motion of a continuum, one may start to write down the whole microscopic data, involving the interactions, which is very
difficult. The kinetic theory, on the other hand, uses the statistical concepts to
handle the practical problems of the microscopic theory.   
Indeed, starting with a Hamiltonian vector field $X_h$ generating the motion of a single particle, and applying the geometric operations such as complete lifts and vertical representatives, one arrives at the momentum-Vlasov equations without referring to \textit{any Poisson bracket or any Hamiltonian functional}, \cite{esen2012geometry}.  Precisely, it has been proved in \cite{esen2012geometry} that the momentum-Vlasov equation \eqref{mom-Vlasov-intro} can be written in the form
\begin{equation}\label{gp}
\dot{\Pi}^{v}=VX_h^{c\ast }\left( z,\Pi \right),   
\end{equation}
where $VX_h^{c\ast }$ is the vertical (evolutionary) representative (see \cite{Olver-93,Saunders-Jets} for the theory of the vertical representatives) of the
complete cotangent lift $X_h^{c\ast }$ of the Hamiltonian vector field $X_h$. Here, $\dot{\Pi}^{v}$ is a vector field defined as the vertical lift of the 1-form $\dot{\Pi}$, see for example \cite{yano}.  

Finally, let us comment on the geometrization \eqref{gp}, with some informal intuitions. A single particle traces a curve in the momentum phase space $T^*Q$ so that its motion is determined by an ordinary differential equation (ODE) with time as the independent variable. This is a Lagrangian submanifold of the iterated tangent bundle $TT^*Q$. On the other hand, the motion of the whole continuum is determined by a partial differential equation (PDE) governing a (scalar or vectorial) field. That is, mathematically, the single particle motion is a submanifold of a tangent bundle, whereas the motion of the continuum is a submanifold a jet bundle. Therefore, in order to find a link between these two motions, one needs to propose a geometric operator taking a tangent vector to a jet bundle element. Further, the final product of this operator  must forget the motion on the base level in order just to concentrate on the field parameters. The geometrization in \eqref{gp} does both of these two tasks simultaneously, relating the base motion to the motion of the field $\Pi$, by removing at the same time the affects of the dynamics on the base level. This geometric pathway approach has already found applications in some fluid theories as well \cite{EsGrGuPa19,esen2011lifts}.

\subsubsection*{\textbf{The second goal of this work. Matched pair decomposition of the momentum-Vlasov equations.}} The second aim of this work is to show that both the Vlasov equation \eqref{Boltzmann-intro} and the momentum-Vlasov equations \eqref{mom-Vlasov-intro} admit matched pair decompositions in the form \eqref{LPEgh-intro}. To establish this goal, we first carry the matched pair decomposition ${\mathfrak{T}\mathcal{Q}}=\G{s}\bowtie\G{n}$ of the symmetric contravariant tensor fields to the functions $\C{F}(T^*\C{Q})$, and then to  $\mathfrak{X}_{\mathrm{\mathrm{Ham}}}(T^*\C{Q})$ by the Lie algebra homomorphisms  \eqref{TQ-to-F} and \eqref{GCCL} respectively. More precisely, the matched pair decomposition of the non-flat functions $\C{F}_0(T^*\C{Q}) \subseteq \C{F}(T^*\C{Q})$ is stated in \eqref{C-decomp-1-}, while the matched pair decomposition of the Hamiltonian vector fields $\mathfrak{X}_{\mathrm{\mathrm{Ham,0}}}(T^*\C{Q})$ of non-flat functions in \eqref{formal-Ham-vf}. As a result, the Vlasov equation and the momentum-Vlasov equations are presented as matched pair Lie-Poisson systems in Corollary \ref{coroll-matched-Vlasov} and Corollary \ref{coroll-matched-mom-Vlasov}, respectively. We further express, in detail, all mutual actions and induced cross terms. 

\subsubsection*{\textbf{Organization}}

The paper is organized as follows. 

In Section \ref{Sec-DMIS} we present the background material. Namely, we shall recall in Subsection \ref{Subsec-mp} the basics of the matched pair Lie algebras (and their double cross sums), and then in Subsection \ref{Subsec-mrd} the Hamiltonian dynamics on double cross sums of Lie algebras. 

On the following three sections, we shall analyse the Lie-Poisson theory of the (plasma) kinetic moments, and the functions and the Hamiltonian vector fields on cotangent bundle. To this end, we reserve Section \ref{sect-LPD} to the Lie-Poisson dynamics of the kinetic moments of Vlasov plasma. More precisely, in Subsection \ref{subsect-symm-contra-tensor} we recall the space of (formal) contravariant tensor fields, while its dual, the space of covariant tensor fields is presented in \ref{subsect-symm-cov-tensors}. Finally, in Subsection \ref{subsect-kinetic-moments-dynamics} we present the (matched) Lie-Poisson equations for kinetic moments.

Next, in Section \ref{sect-mp-Vla} we present the matched pair analysis of the Lie-Poisson realization of the Vlasov plasma. In Subsection \ref{subsect-formal-power-series-Lie-alg} we recall the Borel's theorem to introduce, in the Lie algebra of functions, the subalgebra of non-flat functions (which is isomorphic to the space of formal power series in momentum variables). As a result, we gain the computational advantage of its graded structure. The dual space of this subalgebra is introduced in Subsection \ref{subsect-finite-moments}, and is identified with the functions with finite moments. The decomposition of the Vlasov plasma dynamics, then, is discussed in Subsection \ref{subsect-mp-Vla}.

The momentum-Vlasov dynamics, on the other hand, is studied in Section \ref{mp-mVla-sec}. In Subsection \ref{subsect-Ham-Lie-alg} we introduce the basics on the Lie algebra of Hamiltonian vector fields. Moreover, in view of the canonical projection from the functions to the Hamiltonian vector fields, we introduce the graded subalgebra of the non-flat Hamiltonian vector fields. Then, in Subsection \ref{subsect-1-form-div} we consider the dual space of the Hamiltonian vector fields, and formulate the coadjoint action. Lastly, in Subsection \ref{Sec-m-Vlasov} we present the (decomposition) of the Lie-Poisson equations.

Finally, in Section \ref{sect-can-diff} we study the Lie group counterpart of the decompositions we have considered in the previous sections. To this end, we first recall the fundamentals of the matched pairs of Lie groups, and their double cross products, in Subsection \ref{subsect-matched-Lie-gr}. Then, in Subsection \ref{subsect-decomp-diff-can} we present the double cross product group structure of the (sub)group of canonical diffeomorphisms preserving the canonical 1-form, which may be considered to be the Lie group of the Lie (sub)algebra of non-flat Hamiltonian vector fields.

\subsubsection*{Notations and Conventions}~

Throughout the text $\C{Q}\subseteq \B{R}^3$ will stand for a closed manifold. We shall denote by $\C{F}(\C{Q})$ the set of (real valued) smooth functions on $\C{Q}$. The integration over $\C{Q}$ (resp. $T^\ast\C{Q}$) will be against the top form $dq$ (resp. $dqdp$).

\section{Lie-Poisson dynamics on double cross sum Lie algebras} \label{Sec-DMIS}

In this section, se shall present the material that will be needed in the sequel. More precisely, we shall first recall the matched pairs of Lie algebras and their double cross sums, and then the basics of the Hamiltonian dynamics on double cross sums of Lie algebras.

\subsection{Double cross sum Lie algebras}\label{Subsec-mp}~

\subsubsection*{The matched pair theory}

We begin with a brief review of the matched pair theory for Lie algebras from \cite{Maji90,Majid-book}, see also \cite{LuWein90,Maji90-II,Ta81,Zhan10}. 

A pair $(\G{g},\G{h})$ of Lie algebras, with mutual actions
\begin{align} 
& \rt:\G{h}\ot\G{g}\rightarrow \G{g},\quad \eta\ot \xi \mapsto \eta\rt \xi, \label{Lieact-left}\\ &\lt:\G{h}\ot\G{g}\rightarrow \G{h}, \quad \eta\ot \xi \mapsto \eta\lt \xi \label{Lieact-right}
\end{align}
is called a ``matched pair of Lie algebras'' if the mutual actions and the individual Lie brackets satisfy
\begin{equation} \label{compcon-mpl}
\begin{split}
\eta \rt \lbrack \xi,\xi']=[\eta \rt \xi,\xi']+[\xi,\eta \rt \xi']+(\eta\lt \xi)\rt \xi'-(\eta \lt \xi')\rt \xi, \\
\lbrack \eta,\eta']\lt\xi =[\eta,\eta'\lt\xi ]+[\eta\lt\xi ,\eta']+\eta\lt (\eta' \rt \xi)-\eta'\lt (\eta\rt \xi ),
\end{split}
\end{equation}
for any $\xi,\xi'\in\mathfrak{g}$, and any $\eta,\eta'\in \mathfrak{h}$. In this case, their ``double cross sum'' $\G{g}\bowtie \G{h}:= \G{g}\oplus \G{h}$ becomes a Lie algebra through
\begin{equation}\label{mpla}
\Big[ (\xi,\eta);(\xi',\eta')\Big ]=\Big( [\xi,\xi']+\eta\rt \xi'-\eta'\rt \xi; 
 [\eta,\eta']+\eta\lt \xi'-\eta' \lt \xi \Big)
\end{equation}
for any $(\xi,\eta),(\xi',\eta')\in\G{g}\bowtie \G{h}$. 

\begin{remark}
The matched pair theory may be considered to be a generalization of the well-known semi-direct sum construction. More precisely, in case the left action \eqref{Lieact-left} is trivial, then $\G{g}\bowtie \G{h} = \G{g}\ltimes \G{h}$, and similarly if \eqref{Lieact-right} is trivial, then $\G{g}\bowtie \G{h}=\G{g}\rtimes \G{h}$. Needless to say, if both of the actions \eqref{Lieact-left} and \eqref{Lieact-right} are trivial, then the double cross sum Lie algebra reduces to a direct sum Lie algebra. 
\end{remark}

In the course of a double cross sum realization of a Lie algebra, \cite[Prop. 8.3.2]{Majid-book} is indispensible. As such, we record it here.

\begin{proposition} \label{universal-prop}
Given a Lie algebra $\G{K}$, with two Lie subalgebras $\G{g},\G{h} \subseteq \G{K}$, if $\G{K}\cong\G{g}\oplus\G{h}$ as vector spaces thorough 
\[
\G{g}\oplus \G{h}\to \G{K}, \qquad (\xi,\eta)\mapsto \xi+\eta,
\]
then $(\G{g},\G{h})$ is a matched pair of Lie algebras, and moreover $\G{K}\cong\G{g}\bowtie\G{h}$ as Lie algebras. The mutual actions, in this case, are given by 
\begin{equation} \label{mab-defn}
[\eta,\xi]=\Big(\eta\rt \xi, \eta\lt \xi\Big) \in \G{K},
\end{equation}
for any $\xi\in \G{g}$, and any $\eta\in \G{h}$.
\end{proposition}

\subsubsection*{The coadjoint representation of a double cross sum Lie algebra}

The coadjoint action of a Lie algebra on its dual is the integral part of the Lie-Poisson equation \eqref{LP-eqn}. Accordingly, we devote the present paragraph to a quick review of the coadjoint representation of a double cross sum Lie algebra on its dual, \cite{EsSu16}. To this end, we shall need the following mappings induced by the mutual actions \eqref{Lieact-left} and \eqref{Lieact-right}.

To begin with, the left action \eqref{Lieact-left} of $\G{h}$ on $\G{g}$ induces a right action 
\begin{equation} \label{eta-star}
\overset{\ast }{\lt} :\G{g}^\ast\ot \G{h}\longrightarrow \G{g}^\ast, 
\qquad \langle \mu \overset{\ast }{\lt} \eta, \xi \rangle:=\langle \mu, \eta \rt \xi \rangle,
\end{equation}
for any $\xi\in \G{g}$, any $\mu\in \G{g}^\ast$, and any $\eta\in \G{h}$. Induced also by \eqref{Lieact-left} is the mapping
\begin{equation} \label{b}
\G{b}_\xi: \G{h} \longrightarrow \G{g},\qquad \G{b}_\xi(\eta)=\eta\rt \xi,
\end{equation}
for any $\xi\in\G{g}$, and its transpose 
\begin{equation} \label{b*}
\G{b}_\xi^*:\G{g}^*\longrightarrow \G{h}^*, \qquad \langle \G{b}_\xi^*\mu,\eta \rangle: = \langle \mu, \G{b}_\xi \eta \rangle = \langle \mu, \eta\rt \xi  \rangle.
\end{equation}
Similarly, the right action \eqref{Lieact-right} induces the left action
\begin{equation} \label{xi-star}
\overset{\ast }{\rt}: \G{g}\ot \G{h}^* \to \G{h}^*, \qquad 
\langle \xi \overset{\ast }{\rt}\nu, \eta \rangle :=\langle\nu,  \eta \lt\xi\rangle
\end{equation}
of $\G{g}$ on $\G{h}^\ast$, and the mappings 
\begin{equation} \label{a}
\G{a}_\eta:\G{g}\mapsto \G{h}, \qquad \G{a}_\eta(\xi)=\eta\lt \xi, 
\end{equation}
with its transpose
\begin{equation}\label{a*}
\G{a}_\eta^*:\G{h}^*\mapsto \G{g}^*, \qquad \langle \G{a}_\eta^* \nu,\eta \rangle :=
\langle \nu,\G{a}_\eta \xi \rangle=\langle \nu,\eta\lt \xi \rangle.
\end{equation}
With all these at hand, the coadjoint action of $\G{g}\bowtie \G{h}$ on $\G{g}^\ast\oplus \G{h}^\ast$ is given by
\begin{equation} \label{coad}
\ad_{(\xi,\eta)}^{\ast}(\mu,\nu)=\Big(\ad^{\ast}_{\xi} \mu -\mu \overset{\ast }{\lt}\eta - \mathfrak{a}_{\eta}^{\ast}\nu,   
\ad^{\ast}_{\eta} \nu +\xi \overset{\ast }{\rt}\nu+ \mathfrak{b}%
_{\xi}^{\ast}\mu\Big),
\end{equation}
for any $(\xi,\eta)\in \G{g}\bowtie \G{h}$, and any 
$(\mu,\nu) \in \G{g}^\ast\oplus \G{h}^\ast$.

\subsection{Dynamics on double cross sum Lie algebras}\label{Subsec-mrd}~ 

Let us recall first that given a Lie algebra $\G{g}$, the dual space $\G{g}^\ast$ has the structure of a Poisson manifold, called the Lie-Poisson structure, with respect to which the Poisson Lie bracket on $\C{F}(\G{g}^\ast)$ reduces to the Lie bracket of $\G{g}$ on the linear functions on $\G{g}^\ast$. More precisely,  
\[
\langle \{\C{H}, \C{G}\} , \mu\rangle = - \langle \mu, [\frac{\d \C{H}}{\d \mu},\frac{\d \C{G}}{\d \mu}]\rangle
\]
for any $\C{H}, \C{G}  \in \G{g}^\ast$. Accordingly, the dynamics generated by a Hamiltonian functional $\C{H}:\G{g}^\ast\to \B{R}$ may be formulated as
\begin{equation}\label{Lie-Poisson-eqn}
\frac{d \mu}{dt} = -\ad^\ast_{\frac{\d \C{H}}{\d \mu}}\mu,
\end{equation}
which is called the Lie-Poisson equation. 

In the case of a double cross sum Lie algebra $\G{g}\bowtie\G{h}$, on the other hand, the Lie Poisson equation \eqref{Lie-Poisson-eqn} takes the form of
\begin{equation}\label{Lie-Poisson-eqn-double}
\frac{d (\mu,\nu)}{dt} = -\ad^\ast_{\left(\frac{\d \C{H}}{\d \mu},\frac{\d \C{H}}{\d \nu}\right)}(\mu,\nu),
\end{equation}
where $(\mu,\nu)\in \G{g}^\ast\oplus \G{h}^\ast$, and $\C{H} = \C{H}(\mu,\nu)$. In view of the coadjoint action \eqref{coad} of $\G{g}\bowtie \G{h}$ on $\G{g}^\ast\oplus \G{h}^\ast$, the right hand side of \eqref{Lie-Poisson-eqn-double} reads
\[
\left(\frac{d \mu}{dt}, \frac{d \nu}{dt}\right) = \left(-\ad^\ast_{\frac{\d \C{H}}{\d \mu}} \mu +\mu \overset{\ast }{\lt}\frac{\d \C{H}}{\d \nu} + 
\G{a}_{\frac{\d \C{H}}{\d \nu}}^\ast\nu,   
-\ad^{\ast}_{\frac{\d \C{H}}{\d \nu}} \nu -\frac{\d \C{H}}{\d \mu} \overset{\ast }{\rt}\nu- 
\G{b}_{\frac{\d \C{H}}{\d \mu}}^\ast\mu\right),
\]
or equivalently the ``matched Lie-Poisson equations''
\begin{equation}\label{LPEgh}
\begin{split}
& \underbrace{\frac{d\mu}{dt} = -
\ad^{\ast}_{\frac{\delta\mathcal{H}}{\delta\mu}}\mu}_{\text{Lie-Poisson eqn. on }\G{g}^*}
+
\underbrace{\mu\overset{\ast }{\lt}
\frac{\delta\mathcal{H}}{\delta\nu}}
_{\text{action of } \G{h}}
+
\underbrace{\mathfrak{a}_{\frac{\delta\mathcal{H}}{\delta\nu}}^{\ast}\nu}
_{\text{action of }\G{g}}, 
\\
&\underbrace{\frac{d\nu}{dt} =
-
\ad^{\ast}_{\frac{\delta\mathcal{H}}{\delta\nu}}\nu}_
{\text{Lie-Poisson eqn. on }\G{h}^*}
-
\underbrace{
\frac{\delta\mathcal{H}}{\delta\mu} \overset{\ast }{\rt}\nu}_{\text{action of }\G{g}}
- \underbrace{\mathfrak{b}
_{\frac{\delta\mathcal{H}}{\delta\mu}}^{\ast}\mu
}_{\text{action of }\G{h}}.
\end{split}
\end{equation}

\section{Lie-Poisson dynamics of the kinetic moments}\label{sect-LPD}

In the present section we shall analyse the Lie-Poisson dynamics of the kinetic moments of Vlasov plasma, from the matched pairs point of view. To this end, we shall recall the space of (formal) contravariant tensor fields, and its dual, the space of covariant tensor fields. Then, upon deriving the decomposition of the coadjoint action of contravariant formal tensors on covariant tensors, we shall conclude the (matched) Lie-Poisson equations for kinetic moments.
 
\subsection{Symmetric contravariant formal tensor fields}\label{subsect-symm-contra-tensor}~

Given an $n$-manifold $\mathcal{Q}$ without boundary, the ``space of contravariant formal tensor fields'' is the space denoted by
\[
{\mathfrak{T}\mathcal{Q}}:=\prod _{k\geq 0} \mathfrak{T}
^k\mathcal{Q},
\]
with $\mathfrak{T}^k\mathcal{Q}$ being the space of $k$-th order symmetric
contravariant tensor fields on $\mathcal{Q}$. Accordingly, given a local coordinate system $(q^\ell)$ on $\mathcal{Q}$, an element of $\mathfrak{T}\mathcal{Q}$ may be expressed as a formal sum
\begin{equation} \label{X^n}
\mathbb{X}=\sum_{k\geq 0}\mathbb{X}^{k}=\sum_{k\geq 0}\mathbb{X}^{i_{1}i_{2}...i_{k}}(q)
\partial {q^{i_{1}}}\otimes ...\otimes \partial {q^{i_{k}}},
\end{equation}
where $\mathbb{X}^{i_{1}i_{2}...i_{k}} \in \C{F}(\C{Q})$ are assumed to be symmetric on the indices. 

Let us note in particular that the space $\G{T}^0\C{Q}$ of zeroth order tensors is nothing but the space $\C{F}(\C{Q})$ of smooth functions over $\C{Q}$, and the space $\mathfrak{T}^1\mathcal{Q}$ of first order tensors coincides with  the space $\mathfrak{X}(\mathcal{Q})$ of smooth vector fields on 
$\mathcal{Q}$. 

The space $\mathfrak{T}\mathcal{Q}$ admits the structure of a Lie algebra through the Schouten concomitant which is given by 
\begin{equation} \label{SC-def}
\left[ \mathbb{X}^k,\mathbb{Y}^m\right]:=\Big(k\mathbb{X}^{i_{m+1}...i_{m+k-1}\ell}\mathbb{Y}^{i_{1}...i_{m}}_{,\ell}
-m\mathbb{Y}^{i_{k+1}...i_{k+m-1}\ell}  \mathbb{X}^{i_{1}i_{2}...i_{k}}_{,\ell}\Big)\, \partial {q^{i_{1}}}\otimes ...\otimes \partial {q^{i_{k+m-1}}},
\end{equation}
see for instance \cite{KoMiSl93,Ma97,Sc40}. Here we use the abbreviation $A_{,\ell}$ in order to denote the partial derivative of the object $A$ with respect to $q^\ell$. 

It worths to note that the Schouten concomitant is trivial on $\C{F}(\C{Q})$, and coincides with the Jacobi-Lie bracket of vector fields on $\mathfrak{X}(\mathcal{Q})$. Accordingly,
\begin{equation} \label{s}
\G{s}:=\bigoplus _{k=0}^{1}\mathfrak{T} ^k\mathcal{Q}=\C{F}(\C{Q})  \rtimes \mathfrak{X}(\mathcal{Q})
\end{equation}
is a Lie subalgebra on which the Schouten concomitant takes the form
\begin{equation}\label{pa}
\Big[ (\eta,Z ); (\s,Y ) \Big]
=\Big(Z(\s) -Y(\eta), [ Z,Y] \Big),
\end{equation}
for any $\left(\eta,Z \right),\left(\sigma,Y \right)\in \G{s}$. On the other hand, it follows at once from the graded character of the Schouten bracket that the vector space complement 
\begin{equation}\label{n}
\G{n}:=\prod _{k\geq 2} \mathfrak{T} ^k\mathcal{Q}
\end{equation}
of $\G{s} \subseteq \G{T}\C{Q}$ a also Lie subalgebra. 

As a result, Proposition \ref{universal-prop} yields the double cross sum realization of the Lie algebra $\G{T}\C{Q}$ of symmetric contravariant formal tensor fields. 
 
\begin{proposition} \label{mpdTQ}
The pair $(\G{s},\G{n})$ of Lie subalgebras of ${\mathfrak{T}\mathcal{Q}}$ is a matched pair of Lie algebras, and 
\begin{equation}\label{GTQ-matched-pair}
{\mathfrak{T}\mathcal{Q}} =\mathfrak{s}\bowtie \mathfrak{n}. 
\end{equation}
Furthermore, given any
\[
(\sigma,Y)\in \G{s}, \qquad \mathbf{X} := \sum_{k\geq 2}\mathbb{X}^k \in \G{n},
\]
the mutual actions are given by
\begin{align} 
 &\rt:\G{n}\otimes \G{s}\to \G{s}, \qquad \mathbf{X}\rt(\sigma,Y)=(0,[\mathbb{X}^{2},\sigma]),\label{actions-I} \\
 & \lt: \G{n}\otimes \G{s}\to \G{n}, \qquad  \mathbf{X}\lt(\sigma,Y)=\sum_{k\geq 2} ([\mathbb{X}^{k+1},\sigma]-\C{L}_Y \mathbb{X}^{k}).\label{actions-II}
\end{align}
\end{proposition}

\begin{proof}
In view of \eqref{mab-defn}, we compute the mutual actions via  
\begin{equation} \label{mp-decomp-TQ}
  [\mathbf{X} ,(\sigma,Y)] = \mathbf{X} \rt(\sigma,Y) + \mathbf{X}  \lt(\sigma,Y).
\end{equation}
Accordingly, for any $k\geq 2$, we compute
\begin{align*}
    [\mathbb{X}^{k},(\sigma,Y)]& = [\mathbb{X}^{k},\sigma] + [\mathbb{X}^{k},Y]  = 
    [\mathbb{X}^{k},\sigma] - \C{L}_Y \mathbb{X}^{k},
\end{align*}
where $ \C{L}_Y \mathbb{X}^\textbf{k}$ stands for the Lie derivative of $\mathbb{X}^{k} \in \G{n}$ in the direction of $Y\in \G{X}(\C{Q})$, whereas 
\begin{equation} \label{X-sigma}
[\mathbb{X}^{k},\sigma]=k\mathbb{X}^{i_1\ldots i_{k-1}\ell}\sigma_{,\ell}\partial{q^{i_1}}\odots \partial{q^{i_{k-1}}}.
\end{equation}
As a result of \eqref{mp-decomp-TQ} now, we obtain
\begin{equation}\label{n-on-s}
    \mathbb{X}^{k}\rt(\sigma,Y)= \begin{cases}
                                \displaystyle [\mathbb{X}^2,\sigma], & \mbox{if } k=2, \\
                                0, & \mbox{if } k \geq 3,
                              \end{cases}
\end{equation}
and
\begin{equation}\label{s-on-n}
   \mathbb{X}^{k}\lt(\sigma,Y)= \begin{cases}
                              \displaystyle   -\C{L}_Y \mathbb{X}^{2} , & \mbox{if } k=2, \\[.3cm]
                   \displaystyle             [\mathbb{X}^{k},\sigma] -\C{L}_Y \mathbb{X}^{k}, & \mbox{if } k \geq 3.
                              \end{cases}
\end{equation}
The claim thus follows from the fact that $[\mathbb{X}^k,\sigma] \in \G{T}^{k-1}\C{Q}$, while $\C{L}_Y \mathbb{X}^k \in \G{T}^k\C{Q}$.
\end{proof}

\subsection{Symmetric covariant tensor fields}\label{subsect-symm-cov-tensors}~

Along the way towards the Lie-Poisson dynamics of the kinetic moments of Vlasov plasma, we shall now recall the space $\G{T}^\ast\C{Q}$ of symmetric covariant tensor fields. To be more precise, we shall consider the space 
\[
\G{T}^\ast\C{Q} := \bigoplus_{k\geq 0} \G{T}^\ast_k\C{Q},
\]
where $\G{T}^\ast_k\C{Q}$ being the space of symmetric covariant tensors of order $k$, \cite{GiHoTr08}. An element of $\G{T}^\ast_k\C{Q}$, then, may be denoted by
\[
\B{A}_{i_1\ldots i_m}( q) dq^{i_1}\odots dq^{i_m},
\]
where the coefficients are symmetric on the indices. As such, an element of  $\G{T}^\ast\C{Q}$ is a finite sum
\[
\mathbb{A}=\sum_{k\geq 0}\mathbb{A}_{i_1\ldots i_k}( q) dq^{i_1}\odots dq^{i_k},
\]
and a non-degenerate pairing of $\G{T}\C{Q}$ and $\G{T}^\ast\C{Q}$ may be given by
\begin{equation}\label{pairing}
\left\langle \mathbb{A},\mathbb{X}\right\rangle :=  \sum_{k\geq 0}\left\langle \mathbb{A}_k,\mathbb{X}^k\right\rangle =  \sum_{k\geq 0}\,\int_{\mathcal{Q}}\, \mathbb{X}^k\lrcorner \mathbb{A}_k\, d_nq=  \sum_{k\geq 0}\,\int_{\mathcal{Q}} \,\mathbb{A}_{i_1 \ldots i_k}(q) \mathbb{X}^{i_1\ldots i_k}(q)\, d_nq,
\end{equation}
where $\lrcorner$ denotes the tensor contraction, \cite{CherChenLam-book}. 

Accordingly, the dual spaces of the Lie subalgebras $\G{s}\subseteq \G{T}\C{Q}$ of \eqref{s}, and $\G{n}\subseteq\G{T}\C{Q}$ of \eqref{n} may be given by
\[
\G{s}^{\ast} : = \G{T}^\ast_0\C{Q} \oplus \G{T}^\ast_1\C{Q}, \qquad \G{n}^{\ast} : = \bigoplus_{k\geq 2}\G{T}^\ast_k\C{Q} \label{s-n-duals}
\]
which yields the decomposition 
\begin{equation}\label{GTstarQ-matched-pair}
\G{T}^{\ast }\C{Q} = \G{s}^{\ast} \oplus \G{n}^{\ast}
\end{equation}
given by
\begin{equation} \label{decomp-covariant-dual}
\G{T}^{\ast }\C{Q}\ni\sum_{k\geq 0}\mathbb{A}_{i_1\ldots i_k}( q) dq^{i_1}\odots dq^{i_k} = \Big(\B{A}_0(q)+\B{A}_idq^i\Big)+ \sum_{k\geq 2}\mathbb{A}_{i_1\ldots i_k}( q) dq^{i_1}\odots dq^{i_k} \in \G{s}^\ast  \oplus  \G{n}^\ast.
\end{equation}

We shall need the following contraction formulas in the sequel. Given
\[
\B{A}_m=\B{A}_{i_1i_{2}\ldots i_m}(q) dq^{i_1}\odots dq^{i_m} \in \G{T}^\ast_m\C{Q}, \qquad \B{X}^k = \B{X}^{i_{1}i_{2}...i_{k}}(q) \p{q^{i_1}}\odots \p{q^{i_k}} \in \G{T}^k\C{Q},
\]
we we first note that
\begin{equation}\label{llcorner}
\B{X}^k\lrcorner\B{A}_m := \begin{cases}
\displaystyle \B{X}^{i_{1}i_{2}...i_{k}}(q)\B{A}_{i_1i_{2}\ldots i_m}(q) dq^{i_{k+1}}\odots dq^{i_m} \in \G{T}^\ast_{m-k}\C{Q} & \text{if  } m> k, \\[.3cm]
\displaystyle \B{X}^{i_{1}i_{2}...i_{k}}(q)\B{A}_{i_1i_{2}\ldots i_m}(q)\p{q^{i_{m+1}}}\odots \p{q^{i_k}}  \in \G{T}^{k-m}\C{Q} & \text{if  } k>m. 
\end{cases} 
\end{equation}
We shall also make use of the abbreviation 
\begin{equation} \label{diver}
{\rm div} :\mathfrak{T}^k\mathcal{Q} \longrightarrow \mathfrak{T}^{k-1}\mathcal{Q}, \qquad \mathbb{X}^{i_{1}i_{2}\dots i_{k}}(q)
\partial {q^{i_{1}}}\otimes \dots \otimes \partial {q^{i_{k}}} \mapsto k\mathbb{X}^{\ell i_{2}\dots i_{k}}_{,\ell}(q)
\partial {q^{i_{2}}}\otimes \dots \otimes \partial {q^{i_{k}}},
\end{equation}
for $k>0$, and ${\rm div} \mathbb{X}^{0} := 0$. It worths to note that, \eqref{diver} coincides with the classical divergence of a vector field for $k=1$. 

Let us further introduce the following abbreviations. For $k\geq 0$ and $m+k-1\geq 0$,
\begin{equation}\label{abbri}
\begin{split}
& \B{A}_{m+k-1}\star \B{X}^k := m \B{A}_{i_1\ldots i_{m-1} i_{m+1}\dots i_{m+k}}(q) \B{X}^{i_{m+1} \dots i_{m+k}}_{,i_m}(q) \,dq^{i_1} \odots dq^{i_m} \in \G{T}^\ast_m\C{Q}
\\
& \B{X}^k \ast \B{A}_{m+k-1} := k\B{X}^{i_{m+1} \dots i_{m+k-1} \ell} (q)\B{A}_{i_1 \dots   i_{m+k-1}, \ell}(q)\, dq^{i_1} \odots dq^{i_m}  \in \G{T}^\ast_m\C{Q}.
\end{split}
\end{equation}
Adding up, we introduce 
\begin{equation}\label{Lie-gen}
{\rm L}_{\B{X}^k} \B{A}_{m+k-1} := \B{A}_{m+k-1}\star \B{X}^k+ \B{X}^k\ast \B{A}_{m+k-1} \in \G{T}^\ast_m\C{Q},
\end{equation}
which, for $k=1$, reduces to the Lie derivative of the tensor field $\mathbb{A}_m\in \G{T}^\ast_m\C{Q}$ in the direction of $\mathbb{X}^{1} \in \G{X}(\C{Q})$. With all these at hand, we can now express the coadjoint action of $\G{T}\C{Q}$ on $\G{T}^\ast\C{Q}$.

\begin{proposition}\label{prop-coad}
The coadjoint action of $\B{X}\in\mathfrak{T}\mathcal{Q}$ on $\B{A}\in\mathfrak{T}^{\ast }\mathcal{Q}$ is given by
\begin{equation}\label{Coad-VF-}
\ad^\ast_{\B{X}}\B{A} = \sum_{m\geq 0} \widetilde{\mathbb{A}}_m,
\end{equation}
where
\begin{align}  
&\widetilde{\mathbb{A}}_0 = \sum_{k\geq  1}\, 
  \Big(\mathbb{X}^k \ast \mathbb{A}_{k-1}
+ {\rm div} \mathbb{X}^k  \lrcorner  \mathbb{A}_{k-1}\Big)
, \label{Coad-VF-eq}\\
& \widetilde{\mathbb{A}}_m=\sum_{k\geq 0}\, \Big( {\rm L}_{\mathbb{X}^k} \mathbb{A}_{m+k-1} + {\rm div} \mathbb{X}^k  \lrcorner  \mathbb{A}_{k+m-1}\Big), \qquad m\geq 1. \label{Coad-VF-eq-II}
\end{align}
\end{proposition}

\begin{proof}
For any $\B{X}^k, \B{Y}^m \in \G{T}\C{Q}$, with $m+k\geq 1$, and for any $ \B{A}_{m+k-1} \in \G{T}^\ast\C{Q}$, we see at once that
\begin{align*}
&\Big\langle  \ad^\ast_{\B{X}^k}\B{A}_{m+k-1}, \B{Y}^m \Big\rangle =  \Big\langle \B{A}_{m+k-1}, [\B{Y}^m, \B{X}^k] \Big\rangle = \\
& \int_Q \,\B{A}_{i_1\ldots i_{k+m-1}}(q)\Big(m\B{Y}^{i_{k+1}\ldots i_{k+m-1}\ell}(q) \B{X}^{i_1\ldots i_k}_{,\ell}(q) - k\B{X}^{i_{m+1}\ldots i_{k+m-1}\ell}(q) \B{Y}^{i_1\ldots i_m}_{,\ell}(q) \Big)\,d_nq =
 \\
 &m\int_Q \mathbb{A}_{i_1\ldots i_{k+m-1}}(q)\mathbb{Y}^{i_{k+1}\ldots i_{k+m-1}\ell} (q)\mathbb{X}^{i_1\ldots i_k}_{,\ell}(q)\,dq +
 \\
& k \int_Q \,\Big(\mathbb{A}_{i_1\ldots i_{k+m-1},\ell}(q)\mathbb{X}^{i_{m+1}\ldots i_{k+m-1}\ell}(q) + \mathbb{A}_{i_1\ldots i_{k+m-1}}(q) \mathbb{X}^{i_{m+1}\ldots i_{k+m-1}\ell}_{,\ell}(q)\Big)\mathbb{Y}^{i_1\ldots i_m}(q)\, d_nq =
\\
 &
\Big\langle \mathbb{A}_{m+k-1}\star \mathbb{X}^k,\mathbb{Y}^m \Big\rangle 
+
 \Big\langle \mathbb{X}^k \ast \mathbb{A}_{m+k-1} ,\mathbb{Y}^m \Big\rangle
+
 \Big\langle {\rm div} \mathbb{X}^k  \lrcorner  \mathbb{A}_{k+m-1},\mathbb{Y}^m \Big\rangle,
\end{align*}
from which both \eqref{Coad-VF-eq} and \eqref{Coad-VF-eq-II} follow. Let us note that we used integration by parts on the second equality. 
\end{proof}

\subsection{Lie-Poisson dynamics}\label{subsect-kinetic-moments-dynamics}~

As is well-known, the space $\G{T}^\ast\C{Q}$ of symmetric covariant tensor fields may be considered as the kinetic moments of the plasma density function, \cite{gibbons1981collisionless,
GiHoTr08,gibbons2008vlasov,Tr08}. Moreover, the dynamics of the kinetic moments is a coadjoint flow. Accordingly, the Lie-Poisson equation on $\G{T}^\ast\C{Q}$ is given by
\[
\dot{\mathbb{A}}=-\ad^*_{\frac{\p \C{H}}{\p\B{A}}} \B{A},
\]
where $\C{H}:\G{T}^\ast\C{Q} \to \B{R}$, $\C{H} = \C{H}(\B{A})$, is the Hamiltonian functional generating the motion. More precisely, considering $\partial \C{H} / \p \B{A}_m\in\G{T}^m\C{Q}$, it follows readily from \eqref{Coad-VF-eq} and \eqref{Coad-VF-eq-II} that the Lie-Poisson dynamics on the individual moments may be given by 
\begin{equation}  \label{Coad-VF-eq-}
\begin{split}
&\frac {d \mathbb{A}_0}{dt} = - \sum_{k\geq 1}\, 
 \frac{\d  \C{H}}{\d \B{A}_k}\ast \B{A}_{k-1}
- {\rm div} \frac{\d  \C{H}}{\d  \B{A}_k} \lrcorner  \B{A}_{k-1}
\\
&\frac {d \B{A}_m}{dt}=-\sum_{k\geq 0}\,  {\rm L}_{\frac{\d  \C{H}}{\d  \B{A}_k}} \B{A}_{m+k-1}- {\rm div} \frac{\d  \C{H}}{\d  \B{A}_k}  \lrcorner  \B{A}_{k+m-1}, \qquad m\geq 1.
\end{split}
\end{equation}
These equations may be decomposed further via \eqref{GTQ-matched-pair}, \eqref{GTstarQ-matched-pair}, and \eqref{coad}. Accordingly, we now record the dual actions.

\begin{proposition} \label{prop-dualaction}
The left action \eqref{actions-I} gives rise to the right action
\begin{equation} \label{dualaction-I}
\overset{\ast }{\lt}:\G{s}^*\otimes \G{n}\to\G{s}^*, \qquad (\rho,M)\overset{\ast }{\lt}\mathbf{X}:=(-\mathbb{X}^{2}\ast M - {\rm div}\mathbb{X}^2 M,0)
\end{equation}  
for any ${\rm \bf X}=\sum_{k\geq2}\B{X}^k\in \G{n}$, and any $(\rho,M) \in \G{s}^\ast$. Similarly, given any $(\s,Y)\in \G{s}$, and ${\rm \bf A}=\sum_{k\geq2}\B{A}_k \in \G{n}^\ast$, 
the right action \eqref{actions-II} yields the left action 
\begin{equation} \label{dualaction-II}
\overset{\ast }{\rt}:\G{s}\otimes \G{n}^*\to\G{n}^*, \qquad (\sigma,Y)\overset{\ast }{\rt} {\rm \bf A} = \Big ( \C{L}_Y \mathbb{A}_2+ {\rm div}Y  \mathbb{A}_2 ,  
\sum_{m\geq3} \left( \C{L}_Y \mathbb{A}_m+ {\rm div}Y  \mathbb{A}_m + \mathbb{A}_{m-1}\star \sigma \right )\Big ).
\end{equation}
\end{proposition}

\begin{proof}
The first observation follows directly from
\begin{align*}
&\langle (\rho,M)\overset{\ast }{\lt}{\bf X}, (\s,Y) \rangle = \langle (\rho,M), {\bf X} \rt (\s,Y) \rangle  = \langle (\rho,M), \B{X}^2 \rt (\s,Y) \rangle  =\\
&  \langle M,[\mathbb{X}^2,\s] \rangle = \langle M,\ad_{\mathbb{X}^2} \s \rangle = - \langle \ad^*_{\mathbb{X}^2}M, \s \rangle = -\langle {\rm L}_{\mathbb{X}^2}M +{\rm div}\mathbb{X}^2 M, \s \rangle = \\
& -\langle \mathbb{X}^{2}\ast M + {\rm div}\mathbb{X}^2 \lrcorner M, \s\rangle,
\end{align*}
where on the sixth equality we used \eqref{Coad-VF-eq-II}, and on the last equality we used \eqref{Lie-gen}. As for the latter, keeping \eqref{Coad-VF-eq-II} and \eqref{Lie-gen} in mind, we observe for any $m\geq 2$ that
\begin{align*}
&\langle \sigma \overset{\ast }{\rt}\mathbb{A}_m,  \mathbb{X}^{m+1}\rangle = \langle \mathbb{A}_m,  \mathbb{X}^{m+1}\lt \s \rangle = \langle \mathbb{A}_m,  [\mathbb{X}^{m+1},\s] \rangle =
 \\
&\langle \ad^*_\sigma \mathbb{A}_m,  \mathbb{X}^{m+1} 
\rangle = \langle {\rm L}_\sigma  \mathbb{A}_m+ {\rm div} \sigma \mathbb{A}_m,\mathbb{X}^{m+1} \rangle = \langle\mathbb{A}_m\star \sigma, \mathbb{X}^{m+1} 
\rangle,
\end{align*} 
and that 
\begin{align*}
&\langle Y \overset{\ast }{\rt}\mathbb{A}_m,  \mathbb{X}^m\rangle = \langle \mathbb{A}_m,  \mathbb{X}^m\lt Y \rangle = \langle \mathbb{A}_m, [\mathbb{X}^m, Y] \rangle =\\
&  \langle \ad^*_Y \mathbb{A}_m, \mathbb{X}^m
\rangle = \langle  \C{L}_Y \mathbb{A}_m +{\rm div}Y  \mathbb{A}_m, \mathbb{X}^m\rangle,
\end{align*}
from which the result follows.
\end{proof}

We now proceed into the transpositions of the mappings given, for any $(\s,Y)\in \G{s}$ and ${\rm \bf X}=\sum_{k\geq2}\B{X}^k\in \G{n}$, by 
\begin{align} 
\G{b}_{(\s,Y)}& :\G{n} \to \G{s}, \qquad \G{b}_{(\s,Y)}{\mathbf X} := {\mathbf X} \rt (\s,Y), \label{b-ex}
\\
 \G{a}_{\mathbf X}&:\G{s} \to \G{n}, \qquad \G{a}_{\mathbf X}(\s,Y) := {\mathbf X} \lt (\s,Y). \label{a-ex}
\end{align}

\begin{proposition} \label{cross-act}
Given any $(\s,Y)\in \G{s}$ with ${\rm \bf X}=\sum_{k\geq2}\B{X}^k\in \G{n}$, and any $(\rho,M) \in \G{s}^\ast$ with ${\rm \bf A}=\sum_{k\geq2}\B{A}_k \in \G{n}^\ast$, the transposes of the linear operators \eqref{b-ex} and \eqref{a-ex} read 
\begin{align}
 &\G{b}^\ast_{(\s,Y)}:\G{s}^\ast \to \G{n}^\ast, \qquad \G{b}^\ast_{(\s,Y)}(\rho,M) = M\star \sigma,
 \label{b-ex-dual}
\\
&\G{a}^\ast_{\mathbf X} :\G{n}^\ast \to \G{s}^\ast, \qquad \G{a}^\ast_{\mathbf X}{\mathbf A}
= \left (-\sum_{k\geq 2}\big(\mathbb{X}^{k+1}\ast \mathbb{A}_k + {\rm div}\mathbb{X}^{k+1} \lrcorner \mathbb{A}_k\big),  
- \sum_{k\geq 2} \big ( {\rm L}_{\mathbb{X}^k} \mathbb{A}_k + {\rm div} \mathbb{X}^k  \lrcorner  \mathbb{A}_k \big)\right) , \label{a-ex-dual}  
\end{align}
respectively.
\end{proposition}

\begin{proof}
The former follows at once from
\begin{align*}
& \langle\G{b}^\ast_{(\s,Y)}(\rho,M), {\bf X}\rangle = \langle (\rho,M), \G{b}_{(\s,Y)}{\bf X}\rangle = \langle (\rho,M), {\bf X} \rt (\s,Y)\rangle = \langle (\rho,M), \mathbb{X}^{2} \rt (\s,Y)\rangle = \\
& -\langle M, \ad_\s \mathbb{X}^{2} \rangle=\langle \ad_\s ^*M, \mathbb{X}^{2} \rangle = \langle M\star \s, \mathbb{X}^{2} \rangle = \langle M\star \s, {\bf X} \rangle,
\end{align*}
where we used \eqref{Coad-VF-eq-II} and \eqref{Lie-gen} on the sixth equality. As for the latter, we observe for $k\geq 2$ that
\begin{align*}
&\langle \G{a}^\ast_{\mathbb{X}^{k+1}}\mathbb{A}_k, \s \rangle 
= \langle \mathbb{A}_k, \G{a}_{\mathbb{X}^{k+1}}\s \rangle =\langle \mathbb{A}_k,  \mathbb{X}^{k+1} \lt \s \rangle  =  \\
&  \langle \mathbb{A}_k,  [\mathbb{X}^{k+1},\s] \rangle 
=  \langle -\ad^*_{\mathbb{X}^{k+1}}\mathbb{A}_k,\s\rangle = \langle - \mathbb{X}^{k+1}\ast \mathbb{A}_k - {\rm div}\mathbb{X}^{k+1} \lrcorner \mathbb{A}_k, \s \rangle
\end{align*}
where, this time, we used \eqref{Coad-VF-eq} on the fifth equality, and that 
\begin{align*}
 &\langle \G{a}^\ast_{\mathbb{X}^k}\mathbb{A}_k, Y \rangle 
 = \langle \mathbb{A}_k, \G{a}_{\mathbb{X}^k}Y \rangle 
 =  \langle \mathbb{A}_k, \mathbb{X}^k\lt Y \rangle 
 =\langle \mathbb{A}_k, [\mathbb{X}^k,Y] \rangle 
 = \\ 
& \langle -\ad^*_{\mathbb{X}^k} \mathbb{A}_k, Y
 \rangle = \langle - {\rm L}_{\mathbb{X}^k} \mathbb{A}_k - {\rm div} \mathbb{X}^k  \lrcorner  \mathbb{A}_k, Y \rangle
\end{align*}
in view of \eqref{Coad-VF-eq-II}. The claim thus follows. 
\end{proof}

Finally, we are ready for the decomposition of the coadjoint action of $\G{T}\C{Q}$ on $\G{T}^\ast\C{Q}$.

\begin{proposition} \label{coaddecompAexp}
Given any $(\s,Y)\in \G{s}$ with ${\rm \bf X}=\sum_{k\geq2}\B{X}^k\in \G{n}$, and any $(\rho,M) \in \G{s}^\ast$ with ${\rm \bf A}=\sum_{k\geq2}\B{A}_k \in \G{n}^\ast$, the coadjoint action of $\mathfrak{T}\mathcal{Q}$ on $\mathfrak{T}^{\ast }\mathcal{Q}$ may be presented as
\[
\ad^\ast_{\big((\s,Y); {\mathbf X}\big)}\Big((\rho,M); {\mathbf A}\Big) = \Big((\tilde{\rho},\tilde{M}); \tilde{\mathbf A}\Big),
\]
where
\begin{equation} \label{coad-mp}
\begin{split}
& \widetilde{\rho}=\C{L}_Y\rho+\rho {\rm div}Y +
\mathbb{X}^{2}\ast M + {\rm div}\mathbb{X}^2 \lrcorner M 
+
\sum_{k\geq2} ( \mathbb{X}^{k+1}\ast \mathbb{A}_k + {\rm div}\mathbb{X}^{k+1} \lrcorner \mathbb{A}_k),
\\
& \widetilde{M}= \rho d\s+\C{L}_Y M+{\rm div}Y M +
\sum_{k\geq2} ( {\rm L}_{\mathbb{X}^k} \mathbb{A}_k + {\rm div} \mathbb{X}^k \lrcorner  \mathbb{A}_k ),
\\
& \widetilde{\mathbb{A}}_2= \sum_{k\geq 2}(
{\rm L}_{\mathbb{X}^k} \mathbb{A}_{k+1} + {\rm div} \mathbb{X}^k  \lrcorner  \mathbb{A}_{k+1})
+
\C{L}_Y \mathbb{A}_2 + {\rm div}Y  \mathbb{A}_2 + M\star \sigma,
\\
& \widetilde{\mathbb{A}}_m=
\sum_{k\geq2}(
{\rm L}_{\mathbb{X}^k} \mathbb{A}_{m+k-1} + {\rm div} \mathbb{X}^k  \lrcorner  \mathbb{A}_{k+m-1})
+  \C{L}_Y \mathbb{A}_m+ {\rm div}Y  \mathbb{A}_m + \mathbb{A}_{m-1}\star \sigma, \quad m\geq 3.
\end{split}
\end{equation}
\end{proposition}

\begin{proof}
In view of \eqref{coad} we have
\begin{align*}
& (\tilde{\rho},\tilde{M})=\ad^{\ast}_{(\s,Y)}(\rho, M)-(\rho, M) \overset{\ast }{\lt}{\mathbf X}- \mathfrak{a}_{\mathbf X}^{\ast}{\mathbf A}
\\
&\tilde{\mathbf A} = \sum_{m\geq 2} \tilde{\B{A}}_m  =\ad^{\ast}_{\mathbf X}{\mathbf A}+(\s,Y) \overset{\ast }{\rt}{\mathbf A}+ \mathfrak{b}%
_{(\s,Y)}^{\ast}(\rho,M).
\end{align*}
Accordingly, having Proposition \ref{prop-dualaction} and Proposition \ref{cross-act}, we just need the individual coadjoint representations of $\G{s}$ and $\G{n}$ on their dual spaces. Precisely, it follows from \eqref{Coad-VF-eq} and \eqref{Coad-VF-eq-II} that
\begin{equation}\label{coad-s-on-s*}
\ad^\ast_{(\s,Y)}(\rho,M) = \Big(\C{L}_Y\rho+\rho {\rm div}Y, \rho d\s+\C{L}_Y M+{\rm div}Y M \Big),
\end{equation}
and directly from \eqref{Coad-VF-eq-II} that
\[
 \ad^\ast_{\mathbf X}{\mathbf A} = \sum_{m\geq 2}\,\left(\sum_{k\geq 2}
{\rm L}_{\mathbb{X}^k} \mathbb{A}_{m+k-1} + {\rm div} \mathbb{X}^k  \lrcorner  \mathbb{A}_{k+m-1}\right).
\]
The claim thus follows.
\end{proof}

As a result, we can now conclude the double cross sum decomposition of the Lie-Poisson dynamics of the kinetic moments.

\begin{corollary}
Given a Hamiltonian functional $\C{H}=\C{H}(\rho,M,\textbf{A})$ on $\G{T}^\ast\C{Q}=\G{s}^\ast \oplus \G{n}^\ast$, the Lie-Poisson equations on symmetric covariant tensor fields may be given by
\begin{equation} \label{coad-mp-kinetic-moments}
\begin{split}
\dot{\rho}&=-\C{L}_{\frac{\partial \C{H}}{\partial M}}\rho-\rho {\rm div}\frac{\partial \C{H}}{\partial M}
-
\frac{\partial \C{H}}{\partial \mathbb{A}_{2}}\ast M - {\rm div}\frac{\partial \C{H}}{\partial \mathbb{A}_{2}}\lrcorner  M 
-
\sum_{k\geq 2}\big ( \frac{\partial \C{H}}{\partial \mathbb{A}_{k+1}}\ast \mathbb{A}_k + {\rm div}\frac{\partial \C{H}}{\partial \mathbb{A}_{k+1}} \lrcorner \mathbb{A}_k\big),
\\
\dot{M}&= - \rho d(\frac{\partial \C{H}}{\partial \rho})-\C{L}_{\frac{\partial \C{H}}{\partial M}} M-{\rm div}\frac{\partial \C{H}}{\partial M} M 
-
\sum_{k\geq 2} \big( {\rm L}_{\frac{\partial \C{H}}{\partial \mathbb{A}_k}} \mathbb{A}_k + {\rm div} \frac{\partial \C{H}}{\partial \mathbb{A}_{k+1}} \lrcorner  \mathbb{A}_k \big ),
\\
\dot{\mathbb{A}}_2&=  -
\sum_{k\geq2} \,\big(
{\rm L}_{\frac{\partial \C{H}}{\partial \mathbb{A}_k}} \mathbb{A}_{k+1} + {\rm div} \frac{\partial \C{H}}{\partial \mathbb{A}_{k+1}}  \lrcorner  \mathbb{A}_{k+1}\big)
-
\C{L}_{\frac{\partial \C{H}}{\partial M}}\mathbb{A}_2 - {\rm div}\frac{\partial \C{H}}{\partial M}  \mathbb{A}_2 -M\star \frac{\partial \C{H}}{\partial \rho} ,  
\\
\dot{\mathbb{A}}_m&=-
\sum_{k\geq2} \,\big(
{\rm L}_{\frac{\partial \C{H}}{\partial \mathbb{A}_k}} \mathbb{A}_{m+k-1} + {\rm div} \frac{\partial \C{H}}{\partial \mathbb{A}_k} \lrcorner  \mathbb{A}_{k+m-1}\big)
-
 \C{L}_{\frac{\partial \C{H}}{\partial M}} \mathbb{A}_m - {\rm div}\frac{\partial \C{H}}{\partial M}  \mathbb{A}_m - \mathbb{A}_{m-1}\star \frac{\partial \C{H}}{\partial \rho}, \quad m\geq 3.
\end{split}
\end{equation}
\end{corollary}

\begin{proof}
In view of \eqref{LPEgh} we have
\begin{align*}
&\left(\frac{d \rho}{dt},\frac{d M}{dt}\right)=-\ad^{\ast}_{\left(\frac{\partial \C{H}}{\partial \rho},\frac{\partial \C{H}}{\partial M}\right)}(\rho, M)
+
(\rho, M) \overset{\ast }{\lt}\frac{\partial \C{H}}{\partial \mathbf A}
+ 
\mathfrak{a}_{\frac{\partial \C{H}}{\partial \mathbf A}}^{\ast}{\mathbf A},
\\
&\frac{d \mathbf A}{dt} =-\ad^{\ast}_{\frac{\partial \C{H}}{\partial \mathbf A}}{\mathbf A} -
\left(\frac{\partial \C{H}}{\partial \rho},\frac{\partial \C{H}}{\partial M}\right) \overset{\ast }{\rt}{\mathbf A}
- 
\mathfrak{b}
_{\left(\frac{\partial \C{H}}{\partial \rho},\frac{\partial \C{H}}{\partial M}\right)}^{\ast}(\rho,M).
\end{align*}
The claim then follows from Proposition \ref{coaddecompAexp}.
\end{proof}

\subsubsection*{Lie-Poisson dynamics of the isentropic compressible fluids}

Restricting the Lie-Poisson dynamics of the kinetic moments onto $\G{s}^\ast \subseteq \G{T}^\ast\C{Q}$ we arrive at the dynamics of a compressible isentropic fluid flow. Indeed, setting ${\rm\bf A} = 0$ in \eqref{coad-mp-kinetic-moments} we obtain 
\begin{equation} \label{momEuler}
\begin{split}
\dot{\rho}&=-\C{L}_{\frac{\delta \C{H}}{\delta M}}\rho-\rho {\rm div}(\frac{\delta \C{H}}{\delta M})
\\
\dot{M}&=-\rho d\big(\frac{\delta \C{H}}{\delta \rho}\big)-\C{L}_{\frac{\delta \C{H}}{\delta M}} M-{\rm div}(\frac{\delta \C{H}}{\delta M})M.
\end{split}
\end{equation}
Next, choosing the Hamiltonian functional  as
\[
\C{H}\left( \rho ,M\right) =\frac{1}{2}\int_{\mathcal{Q}}\frac{M^{2}%
}{\rho }d_3q+\int_{\mathcal{Q}}\rho w ( \rho ) d_3q,
\]
which is the total energy of the continuum consisting of a kinetic term and
a potential term with internal energy $w=w\left( \rho \right) $, and setting $\sigma:={\delta \C{H}}/{\delta \rho}$ with $Y:={\delta \C{H}}/{\delta M}$, we compute  the relations 
\begin{equation}
Y^{i}\rho =\delta ^{ij}M_{j}, \qquad \sigma =-\frac{M^{2}}{\rho ^{2}}%
+r\left( \rho \right),   \label{coordtrans}
\end{equation}
where $r\left( \rho \right) =\rho w^{\prime }+w$ is the enthalpy function. Let us note also that the former is the velocity-momentum identification, whereas the second is the Bernoulli's theorem for isentropic fluid flows. So by employing these basics relations into the intermediate system we arrive at the the Euler equations
in standard formulation, that is \eqref{Euler-st}. Finally, the substitution of \eqref{coordtrans} into  \eqref{momEuler} result with the compressible fluid equation
\begin{equation} \label{Euler-st}
\frac{\partial Y}{\partial t}+\left( Y\cdot \nabla \right) Y=\frac{1}{\rho }%
\nabla p, \qquad \dot{\rho}+{\rm div} (\rho Y)=0
\end{equation}
in standard formulation. 

It worths also to note that as a result of the semi-direct sum Lie algebra structure on $\G{s} = \C{F}(\C{Q})\rtimes \G{X}(\C{Q})$, the dynamics (that is, the Lie-Poisson equations \eqref{momEuler}) may be decomposed further within the present double cross sum framework. Indeed, 
\[
\C{F}(\C{Q})\rtimes \G{X}(\C{Q}) \cong \C{F}(\C{Q})\bowtie \G{X}(\C{Q})
\]
so that the right action 
\[
\lt:\G{X}(\C{Q}) \ot \C{F}(\C{Q}) \to \G{X}(\C{Q})
\]
is trivial.

\subsubsection*{Lie-Poisson dynamics of the kinetic moments of order $\geq 2$}

On the other extreme, the restriction of the dynamics of the kinetic moments onto $\G{n}^\ast$ may be thought of the dynamics of the moments of order $\geq 2$. This time, setting $(\rho,M) = (0,0)$ in \eqref{coad-mp-kinetic-moments} we arrive at the Lie-Poisson equations
\begin{equation} \label{subdyn2}
\dot{\mathbb{A}}_m =  -
\sum_{k\geq 2} \,
{\rm L}_{\frac{\partial \C{H}}{\partial \mathbb{A}_k}} \mathbb{A}_{m+k-1} - {\rm div}\frac{\partial \C{H}}{\partial \mathbb{A}_k}  \lrcorner  \mathbb{A}_{k+m-1}, \quad m\geq 2
\end{equation}  
of the kinetic moments of order $\geq 2$, generated by a Hamiltonian $\C{H} = \C{H}({\rm \bf A})$.

Let us note that further decomposition $\G{n} = \G{n}_0 \oplus \G{n}_1$ of $\G{n}$ into two subspaces 
\[
\G{n}_0 := \bigoplus_{k=2}^r\, \G{T}_k\C{Q}, \qquad \G{n}_1 := \bigoplus_{k\geq r+1}\, \G{T}_k\C{Q}
\] 
cannot be studied in the realm of the double cross sum Lie algebras, as $\G{n}_0 \subseteq \G{n}$ is no longer a Lie subalgebra (but a mere subspace), while $\G{n}_1\subseteq \G{n}$ is.

\section{Lie-Poisson dynamics of the Vlasov plasma}\label{sect-mp-Vla}

The present section is about the matched pair analysis of the Lie-Poisson realization of the Vlasov plasma. We shall begin with the Borel's theorem. As a result, we shall be able to pass from the Lie algebra of functions to its graded subalgebra of non-flat functions, which is isomorphic with the space of formal power series in momentum variables. We shall, this way, be able to take the computational advantage of the graded structure. Next, we consider the dual space of non-flat functions, and identify it with the space of functions with finite moments. The matched pair decomposition of the Vlasov plasma dynamics, then, is discussed in Subsection \ref{subsect-mp-Vla}.

\subsection{The Lie algebra of functions}\label{subsect-formal-power-series-Lie-alg}~

We shall now consider the (Poisson, in the terminology of \cite{Fuks-book}) algebra $\C{F}(T^\ast\C{Q})$ of smooth functions on the cotangent bundle $T^\ast\C{Q}$. Adopting a local coordinate system $(q^\ell,p_\ell)$ on $\C{Q}$, the space $\C{F}(T^\ast\C{Q})$ may be endowed with the structure of a Lie algebra via the opposite of the Poisson bracket
\begin{equation}\label{poisson-on-func}
\{h,g\}:= \frac{\p h}{\p q^\ell}\frac {\p g}{\p p_\ell}-\frac{\p g}{\p q^\ell}\frac{\p h}{\p p_\ell}.
\end{equation}
As is remarked in \cite[Sect. 1]{BlAs79}, only the Lie subalgebra $\C{F}_0(T^*\C{Q}):=\C{F}(\C{Q})[p] \subseteq \C{F}(T^\ast\C{Q})$, of functions which are polynomial in momentum variables, contains the quantum mechanical observables. We shall, on the other hand, consider the space $\C{F}_0(T^*\C{Q})$ of functions which are formal power series in momentum variables. It contains $\C{F}_0(T^*\C{Q})$ as a dense subalgebra, and by a slight abuse of notation, we shall henceforth  call it $\C{F}_0(T^*\C{Q})$.

Now, it follows from the Borel's theorem, see for instance \cite[Thm. I.1.3]{MoerdijkReyes-book}, that the Taylor series expansion 
\begin{equation} \label{T0}
T_0:\C{F}(T^\ast\C{Q}) \to \C{F}_0(T^*\C{Q})
\end{equation}
at zero, by partial derivatives with respect to the momentum variables, induces an isomorphism
\[
\C{F}(T^\ast\C{Q}) / m^\infty_{\C{Q}\times \{0\}} \cong \C{F}_0(T^*\C{Q}),
\]
where, $m^\infty_{\C{Q}\times \{0\}}$ is the space of ``flat functions'' consisting of smooth functions whose  partial derivatives of all orders (including the zeroth order) with respect to the momentum variables vanish on $Q\times \{0\}$. As such, we have a short exact sequence
\[
\xymatrix{
0 \ar[r] & m^\infty_{\C{Q}\times \{0\}} \ar[r] & \C{F}(T^\ast\C{Q}) \ar[r]^{T_0} & \C{F}_0(T^*\C{Q}) \ar[r] & 0,
}
\]
where the second map is the inclusion map whereas $T_0$ is the one in \eqref{T0}. Since short exact sequences of projective modules split, and all vector spaces (over fields) are projective, we have 
\begin{equation} \label{Cinf-decomp}
\C{F}(T^\ast\C{Q}) \cong m^\infty_{\C{Q}\times \{0\}} \oplus \C{F}_0(T^*\C{Q}) 
\end{equation}
as vector spaces. Furthermore, $m^\infty_{Q\times \{0\}}$ is a Lie subalgebra of $\C{F}(T^\ast\C{Q})$ via the (Poisson) bracket operation \eqref{poisson-on-func}. More precisely, we have the following.

\begin{proposition} \label{m-ideal}
The space $m^\infty_{Q\times \{0\}} \subseteq \C{F}(T^\ast\C{Q})$ forms an ideal with respect to the Poisson bracket.
\end{proposition}

\begin{proof}
In view of the Poisson bracket formula \eqref{poisson-on-func}, it suffices to show that $\p h / \p q^i \in m^\infty_{Q\times \{0\}}$, $1 \leq i \leq n$, given $h \in m^\infty_{Q\times \{0\}}$. Accordingly, on the contrary, let
\begin{equation}\label{der-f}
\left.\frac{\p^{s +1} h}{\p q^i\p p_{j_1} \ldots \p p_{j_s}}\right|_{(q,0)} \neq 0,
\end{equation}
for some $s\geq 0$, or more precisely let
\begin{equation}\label{del-f-del-qp}
\frac{\p^{m_1 + \ldots + m_k +1} h}{\p q^i\p p^{m_1}_{j_1} \ldots \p p^{m_s}_{j_s}} = g(q) + r(q,p),
\end{equation}
where $r(q,0) = 0$, $m_1+m_2+...+m_s=s$, and $g(q)\neq 0$. Then, integrating \eqref{del-f-del-qp} with respect to the momentum variables $k$-many times, we arrive at
\[
\frac{\p h}{\p q^i} = \frac{g(q)p_{j_1}\dots p_{j_s}}{m_1! \dots m_s !} + \int\,r(q,p)(dp_{j_1})^{m_1} \dots (dp_{j_s})^{m_s},
\]
and hence
\[
h(q,p) = \frac{p_{j_1}\ldots p_{j_s}}{m_1! \dots m_s !}\,\int\,g(q)dq^i + \int\,r(q,p)(dp_{j_1})^{m_1} \dots (dp_{j_s})^{m_s}dq^i,
\]
which contradicts with $h \in m^\infty_{Q\times \{0\}}$.
\end{proof}

As a result, in view of Proposition \ref{universal-prop}, we have $\C{F}(T^\ast\C{Q}) \cong m^\infty_{\C{Q}\times \{0\}} \bowtie \C{F}_0(T^*\C{Q})$. Moreover, since $m^\infty_{\C{Q}\times \{0\}} \subseteq \C{F}(T^\ast\C{Q})$ is an ideal, the action of $m^\infty_{\C{Q}\times \{0\}}$ on $\C{F}_0(T^*\C{Q})$ is trivial. That is,
\begin{equation}\label{FTQ-isom-FQp}
\C{F}(T^\ast\C{Q}) \cong m^\infty_{\C{Q}\times \{0\}} \rtimes \C{F}_0(T^*\C{Q}).
\end{equation}

We shall conclude the present paragraph with the double cross sum decomposition of the Poisson algebra $\C{F}(T^\ast\C{Q})$ through those of $\C{F}_0(T^*\C{Q})$ and $\G{T}\C{Q}$.

\begin{proposition}\label{prop-SC-Poisson-identify}
The map
\begin{equation}\label{TQ-to-F}
\kappa:{\G{T}\C{Q}}\to \C{F}_0(T^*\C{Q}), \qquad   \mathbb{X}^{i_{1}i_{2}...i_{k}}(q)
\partial {q^{i_{1}}}\otimes ...\otimes \partial {q^{i_{k}}}\mapsto  \B{X}^{i_1\ldots i_k}(q)\,p_{i_1}\ldots p_{i_k}
\end{equation}
is a Lie algebra isomorphism.
\end{proposition}

\begin{proof}
It suffices to see that
\begin{align*}
 & -\{\kappa(\B{X}^k),\, \kappa(\B{Y}^m) \} 
 = -\left\{\B{X}^{i_1\ldots i_k}(q)\,p_{i_1}\ldots p_{i_k},\, \B{Y}^{j_1\ldots j_m}(q)\,p_{j_1}\ldots p_{j_m}\right\} = \\
& -\frac{\p \B{X}^{i_1\ldots i_k}(q)\,p_{i_1}\ldots p_{i_k}}{\p q^\ell} \frac{\p\B{Y}^{j_1\ldots j_m}(q)\,p_{j_1}\ldots p_{j_m}}{\p p_\ell} + \frac{\p\B{Y}^{j_1\ldots j_m}(q)\,p_{j_1}\ldots p_{j_m}}{\p q^\ell}\frac{\p\B{X}^{i_1\ldots i_k}(q)\,p_{i_1}\ldots p_{i_k}}{\p p_\ell} =  \\
& -m {\B{X}^{i_1\ldots i_k}_{,\ell}}(q)\, p_{i_1}\ldots p_{i_k}\,\B{Y}^{j_1\ldots j_{m-1}\ell}(q)\,p_{j_1}\ldots {\hat{p}_\ell}\ldots p_{j_m}\\ 
& \hspace{4cm} +k{\B{Y}^{j_1\ldots j_m}_{,\ell}}(q)\, p_{j_1}\ldots p_{j_m}\,\B{X}^{i_1\ldots i_{k-1}\ell}(q)\,p_{i_1}\ldots {\hat{p}_\ell}\ldots p_{i_k} =\\
& \Big(-m {\B{X}^{i_1\ldots i_k}_{,\ell}}(q)\B{Y}^{i_{k+1}\ldots i_{k+m-1}\ell}(q) + k{\B{Y}^{i_1\ldots i_m}_{,\ell}}(q)\B{X}^{i_{m+1}\ldots i_{m+k-1}\ell}(q)\Big) \,p_{i_1}\ldots  p_{i_{m+k-1}} = \\
& \kappa([\mathbb{X}^k,\B{Y}^m]),
\end{align*}
from which the claim follows.
\end{proof}

Accordingly, we have 
\begin{equation} \label{C-decomp-1-}
\C{F}_0(T^*\C{Q}) \cong \hat{\G{s}} \bowtie \hat{\G{n}},
\end{equation} 
where
\begin{align}
&\hat{\G{s}}=\left\{\hat{\s}=\s+Y\mid\hat{\s}(q,p):=\s(q)+Y^\ell(q)p_\ell\quad\s,Y^\ell \in \C{F}(\C{Q})\right\},\label{hams}
\\ 
&\hat{\G{n}} =\Big\{ \hat{\B{X}}=\sum_{k\geq 2}\hat{\B{X}}^k\mid \hat{\B{X}}(q,p):=\sum_{k\geq2} \B{X}^{i_1\dots i_k}(q)\,p_{i_1}\dots p_{i_k} \quad \B{X}^{i_1\dots i_k} \in \C{F}(\C{Q})  \Big\}. \label{mams}
\end{align} 
More precisely, $\C{F}_0(T^*\C{Q})\ni h \mapsto h^{\hat{\G{s}}} + h^{\hat{\G{n}}} \in \hat{\G{s}}\oplus \hat{\G{n}}$, where 
\[
h^{\hat{\G{s}}}(q,p) + h^{\hat{\G{n}}}(q,p) := \Big(h(q,p) + \frac{\p h(q,0)}{\p p_\ell}p_\ell\Big) + \sum_{k\geq 2}\frac{1}{k!}\frac{\p^k h(q,0)}{\p p_{i_1}\ldots \p p_{i_k}} p_{i_1}\ldots p_{i_k}.
\]
Furthermore,
\begin{equation}\label{power-series-in-p}
\C{F}(T^\ast\C{Q}) \cong m^\infty_{\C{Q}\times \{0\}} \rtimes (\hat{\G{s}} \bowtie \hat{\G{n}}).
\end{equation}
Let us finally record the mutual actions of the matched pair $(\hat{\G{s}} , \hat{\G{n}})$. In view of
\[
[\B{X}^{i_1\dots i_k}(q)\,p_{i_1}\dots p_{i_k} , \s(q)] = k \s_{,\ell}(q)\B{X}^{i_1\dots i_{k-1}\ell}(q)\,p_{i_1}\dots p_{i_{k-1}}
\]
and
\[
[\B{X}^{i_1\dots i_k}(q)\,p_{i_1}\dots p_{i_k} , Y^j(q)p_j] = \Big(-\B{X}^{i_1\dots i_k}_{,\ell}(q)Y^\ell(q) + k\B{X}^{i_1\dots i_{k-1}\ell}(q)Y^{i_k}_{,\ell}(q)\Big)\,p_{i_1}\dots p_{i_k},
\]
for any $\hat{\s}\in \hat{\G{s}}$, and any $\hat{\B{X}}\in \hat{\G{n}}$, we have
\begin{equation}\label{left-action-FQpp}
\B{X} \rt \hat{\s}  = 
2\s_{,\ell}(q)\B{X}^{i\ell}(q)\,p_{i} 
\end{equation}
and
\begin{equation}\label{right-action-FQpp}
\B{X} \lt \hat{\s}  = \sum_{k\geq 2}\Big(-\B{X}^{i_1\dots i_k}_{,\ell}(q)Y^\ell(q) + k\B{X}^{i_1\dots i_{k-1}\ell}(q)Y^{i_k}_{,\ell}(q) + (k+1) \s_{,\ell}(q)\B{X}^{i_1\dots i_k\ell}(q)\Big)\,p_{i_1}\dots p_{i_k}.
\end{equation}

\subsection{The space of functions with finite moments}\label{subsect-finite-moments}~

Identifying the dual space $\C{F}^\ast(T^\ast\C{Q}) = Den(T^\ast\C{Q})$ with $\C{F}(T^\ast\C{Q})$ itself, in accordance with \eqref{pairing} we may identify $\G{T}^\ast\C{Q}$ by 
\[
\C{F}^\ast_0(T^*\C{Q}) :=\{f\in \C{F}(T^\ast\C{Q})\mid \int_{T^\ast\C{Q}}\,p_{i_1}\ldots p_{i_k}f(q,p)dp = 0, \quad \forall\,k\geq n = n(f)\}.
\]
More precisely, \eqref{TQ-to-F} transposes into
\[
\kappa^\ast:\C{F}^\ast_0(T^*\C{Q}) \to \G{T}^\ast\C{Q}, \qquad f\mapsto \sum_{k\geq 0}\left(\int_{T^\ast\C{Q}}\,p_{i_1}\ldots p_{i_k}f(q,p)dp\right)dq^{i_1}\odots dq^{i_k}.
\]
As a result, the decomposition \eqref{decomp-covariant-dual} can readily be put onto $\C{F}^\ast_0(T^*\C{Q})$ as
\begin{equation}\label{F_0-dual-decomp}
\C{F}^\ast_0(T^*\C{Q}) \cong \hat{\G{s}}^\ast\oplus \hat{\G{n}}^\ast,
\end{equation}
where 
\begin{align*}
& \hat{\G{s}}^\ast := \{h\in \C{F}_c(T^\ast\C{Q})\mid \int_{T^\ast\C{Q}}\,h(q,p)p^k\,dp = 0, \quad \forall\,k\geq 2\}, \\ 
&\hat{\G{n}}^\ast := \{h\in \C{F}_c(T^\ast\C{Q})\mid \int_{T^\ast\C{Q}}\,h(q,p)dp = \int_{T^\ast\C{Q}}\, h(q,p)p\,dp =0\}.
\end{align*}

More precisely, given any $f\in \C{F}^\ast_0(T^*\C{Q})$, with say 
\[
\int_{T^\ast\C{Q}}\,f(q,p)p^k\,dp = 0, \qquad k\geq m+1,
\]
setting
\[
\widetilde{f}:= -\frac{\p(p_\ell f)}{\p p_\ell},
\]
$f_0:=f$, and for $k\geq 0$, $f_{k+1}:= \widetilde{f_k} - kf_k$, we have
\[
f_{(m)} := \frac{1}{m!}f_m, \qquad \int_{T^\ast\C{Q}}\,f_{(m)}(q,p)p^k\,dp = \begin{cases}
\displaystyle \int_{T^\ast\C{Q}}\,f(q,p)p^m\,dp & \text{if  } k = m, \\
0 & \text{if  } k\neq m.
\end{cases}
\]
As such, $f - f_{(m)} \in \C{F}^\ast_0(T^*\C{Q}) $ has at most $(m-1)$th moment\footnote[1]{To be more precise, \[
\int_{T^\ast\C{Q}}\,(f(q,p)-f_{(m)}(q,p))p^{m-1}\,dp = (m-1)!\,\int_{T^\ast\C{Q}}\,f(q,p)p^{m-1}\,dp. 
\]}, and we may similarly construct $f_{(m-1)} \in \C{F}^\ast_0(T^*\C{Q})$ so that
\[
\int_{T^\ast\C{Q}}\,f_{(m-1)}(q,p)p^k\,dp = \begin{cases}
\displaystyle\int_{T^\ast\C{Q}}\,f(q,p)p^{m-1}\,dp & \text{if  } k = m-1, \\
0 & \text{if  } k\neq m-1.
\end{cases}
\]
Accordingly, we have the moment decomposition
\begin{equation}\label{h-decomp}
\C{F}^\ast_0(T^*\C{Q})\ni f \mapsto f_\G{s} + f_\G{n}:= (f_{(0)} + f_{(1)}) + \sum_{k=2}^m\, f_{(k)} \in \hat{\G{s}}^\ast \oplus \hat{\G{n}}^\ast.
\end{equation}

We shall conclude the present paragraph with the coadjoint action of $\C{F}_0(T^*\C{Q})$ on $\C{F}^\ast_0(T^*\C{Q})$.

\begin{proposition}\label{prop-coad-II}
Given any $h\in\C{F}_0(T^*\C{Q})$ and $f\in \C{F}^\ast_0(T^*\C{Q})$, the coadjoint action of $\C{F}_0(T^*\C{Q})$ on $\C{F}^\ast_0(T^*\C{Q})$ is given by
\begin{equation}\label{Coad-VF-}
\ad^\ast_hf = \{f,h\}.
\end{equation}
\end{proposition}

\begin{proof}
For any $g \in \C{F}_0(T^*\C{Q})$, we see at once that
\begin{align}\label{coad-F-on-f}
\begin{split}
&\langle \ad^*_h f, g\rangle = \langle f, -\{g,h\}\rangle = \int_{T^\ast\C{Q}}\, f\left(-\frac{\p g}{\p q}\frac{\p h}{\p p} + \frac{\p g}{\p p}\frac{\p h}{\p q}\right)\,dqdp = \\
&   \int_{T^\ast\C{Q}}\, \left(\frac{\p f}{\p q}\frac{\p h}{\p p} + f\frac{\p^2 h}{\p q\p p} \right)g\,dqdp - \int_{T^\ast\C{Q}}\, \left(\frac{\p f}{\p p}\frac{\p h}{\p q} + f\frac{\p^2 h}{\p p\p q}\right)g\,dqdp = \\
& \int_{T^\ast\C{Q}}\, \left(-\frac{\p f}{\p p}\frac{\p h}{\p q} +\frac{\p f}{\p q}\frac{\p h}{\p p}\right)g\,dqdp = \langle \{f,h\},g \rangle.
\end{split}
\end{align}
\end{proof}

In particular, for  
\begin{equation}\label{total-energy}
h(q,p)=\frac{p^2}{2m}+e\phi(q)
\end{equation}
the Lie-Poisson equation 
\[
\frac{d f}{dt}=-\ad^*_{h} f = \{h,f\}
\]
turns out to be the Vlasov equation
\begin{equation}\label{Vlasov-pre}
\frac{d f}{dt}= e \frac{\partial \phi}{\partial q^\ell}   \frac{\partial f}{\partial p_\ell} -\frac{1}{m} \d^{ij}p_i \frac{\partial f}{\partial q^j}.
\end{equation}

\subsection{Vlasov plasma dynamics}\label{subsect-mp-Vla}~

Having equipped with the anti-isomorphism \eqref{TQ-to-F}, we may now pull the Lie-Poisson dynamics of the kinetics moments on the space $\C{F}^\ast_0(T^*\C{Q})$ of functions with finite moments. We shall, however, derive the analogues of Proposition \ref{prop-dualaction}, Proposition \ref{cross-act}, and Proposition \ref{coaddecompAexp}. 

\begin{proposition} \label{prop-dualaction-FQpp}
The left action \eqref{left-action-FQpp} gives rise to the right action
\begin{equation} \label{dualaction-I-FQpp}
\overset{\ast }{\lt}:\hat{\G{s}}^*\otimes \hat{\G{n}}\to \hat{\G{s}}^*, \qquad f_\G{s}\overset{\ast }{\lt}\hat{\B{X}}:=\{\hat{\B{X}},f_\G{s}\}
\end{equation}  
for any $\hat{\B{X}}= \sum_{k\geq2} \hat{\B{X}}^k:=\sum_{k\geq2} \B{X}^{i_1\ldots i_k}(q)p_{i_1}\ldots p_{i_k}\in \hat{\G{n}}$, and any $f_\G{s} \in \hat{\G{s}}^\ast$. Similarly, for any $\hat{\s} \in \hat{\G{s}}$ given by $\hat{\s}(q,p) :=\s(q) + Y^\ell(q)p_\ell$, and any $f_\G{n} \in \hat{\G{n}}^\ast$, 
the right action \eqref{right-action-FQpp} yields the left action 
\begin{equation} \label{dualaction-II-FQpp}
\overset{\ast }{\rt}:\hat{\G{s}}\otimes \hat{\G{n}}^*\to\hat{\G{n}}^*, \qquad \hat{\s}\overset{\ast }{\rt} f_\G{n} = \{f_\G{n},\hat{\s}\}.
\end{equation}
\end{proposition}

\begin{proof}
The first claim \eqref{dualaction-I-FQpp} follows directly from 
\[
\langle f_\G{s}\overset{\ast }{\lt}\hat{\B{X}}, \hat{\s} \rangle = \left\langle f_\G{s}, \hat{\B{X}} \rt \hat{\s} \right\rangle  =  \langle f_\G{s}, -\{\hat{\B{X}},\hat{\s}\} \rangle = \langle -\{f_\G{s},\hat{\B{X}}\}, \hat{\s} \rangle,
\] 
where, on the second equality we used the vanishing of the moments of $f_\G{s} \in \hat{\G{s}}^\ast$ of order $\geq 2$, and on the third equality we used the integration by parts; see for instance \eqref{coad-F-on-f}. 

As for \eqref{dualaction-II-FQpp}, we observe that
\[
\langle \hat{\s}\overset{\ast }{\rt}f_\G{n}, \hat{\B{X}} \rangle = \left\langle f_\G{n}, \hat{\B{X}} \lt \hat{\s} \right\rangle  =  \langle f_\G{n},-\{\hat{\B{X}},\hat{\s}\} \rangle = \langle -\{\hat{\s},f_\G{n}\}, \hat{\B{X}} \rangle,
\] 
from which the result follows.
\end{proof}

\begin{remark}
Let us note that $\{\hat{\B{X}}, f_\G{s}\} \in \hat{\G{s}}^\ast$, for any $\hat{\B{X}}\in \hat{\G{n}}$ and any $f_\G{s}\in\hat{\G{s}}^\ast$. Indeed, for any $k\geq 2$
\[
\int_{T^\ast\C{Q}}\,\{f_\G{s},\hat{\B{X}}\}p_{i_1}\ldots p_{i_k}\,dp =\int_{T^\ast\C{Q}}\,\left(\frac{\p f_\G{s}}{\p q^\ell}\frac{\p \hat{\B{X}}}{\p p_\ell} - \frac{\p \hat{\B{X}}}{\p q^\ell}\frac{\p f_\G{s}}{\p p_\ell}\right)p_{i_1}\ldots p_{i_k}\,dp = 0
\]
as
\begin{align*}
&\int_{T^\ast\C{Q}}\,\frac{\p f_\G{s}}{\p q^\ell}\frac{\p \hat{\B{X}}}{\p p_\ell} p_{i_1}\ldots p_{i_k}\,dp = \sum_{m\geq 2} \,\int_{T^\ast\C{Q}}\,m\frac{\p f_\G{s}}{\p q^\ell}\B{X}^{j_1\ldots j_{m-1}\ell}(q) p_{j_1}\ldots p_{j_{m-1}}p_{i_1}\ldots p_{i_k}\,dp = \\
& \sum_{m\geq 2} \,\frac{\p }{\p q^\ell}\,\int_{T^\ast\C{Q}}\,m f_\G{s} \B{X}^{j_1\ldots j_{m-1}\ell}(q) p_{j_1}\ldots p_{j_{m-1}}p_{i_1}\ldots p_{i_k}\,dp - 
\\
& \hspace{4cm}\sum_{m\geq 2} \,\int_{T^\ast\C{Q}}\,m f_\G{s} \frac{\p \B{X}^{j_1\ldots j_{m-1}\ell}(q)}{\p q^\ell} p_{j_1}\ldots p_{j_{m-1}}p_{i_1}\ldots p_{i_k}\,dp = 0,
\end{align*}
and
\[
\int_{T^\ast\C{Q}}\,\frac{\p \hat{\B{X}}}{\p q^\ell}\frac{\p f_\G{s}}{\p p_\ell}p_{i_1}\ldots p_{i_k}\,dp = - \sum_{m\geq 2} \,\int_{T^\ast\C{Q}}\,mf_\G{s}  \B{X}^{j_1\ldots j_{m-1}\ell}_{,\ell}(q) p_{j_1}\ldots p_{j_{m-1}}p_{i_1}\ldots p_{i_{k-1}}\,dp = 0. 
\]
Similarly, $\{f_\G{n},\hat{\s}\} \in \hat{\G{n}}^\ast$, for any $\hat{s}\in \hat{\G{s}}$, and any $f_\G{n}\in \hat{\G{n}}^\ast$.
\end{remark}

In view of Proposition \ref{prop-coad-II} and \eqref{coad}, the transposes of the mappings 
\begin{align} 
&\G{b}_{\hat{\s}} :\hat{\G{n}} \to \hat{\G{s}}, \quad \G{b}_{\hat{\s}}\hat{\B{X}} = \hat{\B{X}} \rt {\hat{\s}}, \qquad\G{a}_{\hat{\B{X}}}:\hat{\G{s}} \to \hat{\G{n}}, \quad \G{a}_{\hat{\B{X}}}{\hat{\s}} = \hat{\B{X}} \lt {\hat{\s}}, \label{a-b-ex-FQpp} \\
& \ad_{\hat{\s}}:\hat{\G{s}}\to \hat{\G{s}}, \quad \ad_{\hat{\s}}\hat{\lambda} = -\{{\hat{\s}},\hat{\lambda}\}, \qquad \ad_{\hat{\B{X}}}:\hat{\G{n}}\to \hat{\G{n}}, \quad \ad_{\hat{\B{X}}}\hat{\B{Y}} = -\{\hat{\B{X}},\hat{\B{Y}}\} \label{ad-FQpp}
\end{align}
may be obtained from
\[
-\{{\hat{\s}},f_{\G{s}}\}=\ad^\ast_{\hat{\s}}f_{\G{s}} + \G{b}^\ast_{\hat{\s}}f_{\G{s}} \in \hat{\G{s}}^\ast \oplus \hat{\G{n}}^\ast, \qquad -\{\hat{\B{X}},f_{\G{n}}\} = -\G{a}^\ast_{\hat{\B{X}}}f_{\G{n}} + \ad^\ast_{\hat{\B{X}}}f_{\G{n}} \in \hat{\G{s}}^\ast \oplus \hat{\G{n}}^\ast
\]
for any ${\hat{\s}}\in \hat{\G{s}}$, any $\hat{\B{X}}\in \hat{\G{n}}$, any $f_{\G{s}}\in \hat{\G{s}}^\ast$, and any $f_{\G{n}}\in \hat{\G{n}}^\ast$.

\begin{proposition} \label{cross-act-FQpp}
Let $(\hat{\s},\hat{\B{X}})\in\hat{\G{s}}\bowtie \hat{\G{n}}$, and let also $(f_\G{s},f_\G{n})\in \hat{\G{s}}^\ast\oplus \hat{\G{n}}^\ast$. Then, the transposes of the linear operators in \eqref{a-b-ex-FQpp} and \eqref{ad-FQpp} may be given by
\begin{align}
 & \ad^\ast_{\hat{\s}} : \hat{\G{s}}^\ast \to \hat{\G{s}}^\ast,\quad \ad^\ast_{\hat{\s}}f_\G{s} = \{f_{(0)},{\hat{\s}}\} + \{ f_{(1)},Y\}, \qquad\G{b}^\ast_{\hat{\s}}:\hat{\G{s}}^\ast \to \hat{\G{n}}^\ast, \quad \G{b}^\ast_{\hat{\s}}f_\G{s} = \{f_{(1)},\s\},
 \label{ad-b-ex-dual-FQpp}
\\
&\ad^\ast_{\hat{\B{X}}} : \hat{\G{n}}^\ast \to \hat{\G{n}}^\ast,\quad \ad^\ast_{\hat{\B{X}}}f_\G{n} = \sum_{k\geq 2}\{f_{(\ell)},\hat{\B{X}}^{\ell+k}\}, \qquad\G{a}^\ast_{\hat{\B{X}}} :\hat{\G{n}}^\ast \to \hat{\G{s}}^\ast, \quad \G{a}^\ast_{\hat{\B{X}}}f_\G{n}= \{\hat{\B{X}}^k,f_{(k)}\} +\{\hat{\B{X}}^{k+1},f_{(k)}\} , \label{ad-a-ex-dual-FQpp}  
\end{align}
respectively.
\end{proposition}

\begin{proof}
The first claim follows, along the lines of \eqref{h-decomp}, from
\begin{equation}\label{F-f}
\{\hat{\s},f_\G{s}\} = \Big(\{\hat{\s}, f_{(0)}\} + \{Y, f_{(1)}\}\Big) + \{\s, f_{(1)}\} \in \hat{\G{s}}^\ast \oplus \hat{\G{n}}^\ast.
\end{equation}
The latter, on the other hand, follows from
\begin{equation}\label{G-g}
\{\hat{\B{X}},f_\G{n}\} = \Big(\{\hat{\B{X}}^k, f_{(k)}\} +\{\hat{\B{X}}^{k+1}, f_{(k)}\}\Big)+ \sum_{k\geq 2}\{\hat{\B{X}}^{\ell+k}, f_{(\ell)}\} \in \hat{\G{s}}^\ast \oplus \hat{\G{n}}^\ast.
\end{equation}
\end{proof}

The coadjoint action of $\C{F}_0(T^*\C{Q})$ on $\C{F}^\ast_0(T^*\C{Q})$ now decomposes at once.

\begin{corollary} \label{coaddecompAexp-FQpp}
Given any $(\hat{\s},\hat{\B{X}})\in \hat{\G{s}}\bowtie \hat{\G{n}}$, and any $(f_\G{s},f_\G{n})\in \hat{\G{s}}^\ast\oplus \hat{\G{n}}^\ast$, the coadjoint action of $\C{F}_0(T^*\C{Q})$ on $\C{F}^\ast_0(T^*\C{Q})$ is given by
\begin{align}\label{Coad-VF-II}
\begin{split}
& \ad^\ast_{(\hat{\s}+ \hat{\B{X}})}(f_\G{s}+f_\G{n}) = \Big(\{ f_{(0)},\hat{\s}\} + \{ f_{(1)},Y\} + \{f_\G{s},\hat{\B{X}}\} + \{f_{(k)},\hat{\B{X}}^k\} +\{f_{(k)},\hat{\B{X}}^{k+1}\}\Big) + \\
& \hspace{6cm} \Big(\sum_{k\geq 2}\{ f_{(\ell)},\hat{\B{X}}^{\ell+k}\}+ \{f_\G{n},\hat{\s}\}+\{ f_{(1)},\s\}\Big) \in \hat{\G{s}}^\ast\oplus \hat{\G{n}}^\ast.
\end{split}
\end{align}
\end{corollary}

The double cross sum decomposition of the Lie-Poisson dynamics follows as well.

\begin{corollary}\label{coroll-matched-Vlasov}
The Lie-Poisson equations on $\C{F}^\ast_0(T^*\C{Q})=\hat{\G{s}}^\ast \oplus \hat{\G{n}}^\ast$, generated by a Hamiltonian functional $\C{H}=\C{H}(f_\G{s},f_\G{n})$, may be given by
\begin{equation} \label{coad-mp-TQpp}
\begin{split}
& \frac{d f_\G{s}}{dt}=\{\frac{\partial \C{H}}{\partial f_\G{s}}, f_{(0)}\} + \{\frac{\partial \C{H}}{\partial f_{(1)}}, f_{(1)}\} + \{ \frac{\partial \C{H}}{\partial f_\G{n}},f_\G{s}\} + \{\frac{\partial \C{H}}{\partial f_{(k)}},f_{(k)}\} +\{\frac{\partial \C{H}}{\partial f_{(k+1)}},f_{(k)}\} ,
\\
& \frac{d f_\G{n}}{dt}= \sum_{k\geq 2}\{\frac{\partial \C{H}}{\partial f_{(\ell+k)}},f_{(\ell)}\} + \{ \frac{\partial \C{H}}{\partial f_\G{s}},f_\G{n}\} + \{ \frac{\partial \C{H}}{\partial f_{(0)}}, f_{(1)}\}.
\end{split}
\end{equation}
\end{corollary}

\begin{proof}
Once again, along the lines of \eqref{LPEgh} we have
\begin{align*}
&\frac{d f_\G{s}}{dt}=-\ad^{\ast}_{\frac{\partial \C{H}}{\partial f_\G{s}}} f_\G{s} + f_\G{s} \overset{\ast }{\lt}\frac{\partial \C{H}}{\partial f_\G{n}} + 
\mathfrak{a}_{\frac{\partial \C{H}}{\p f_\G{n}}}^{\ast}{ f_\G{n}},
\\
&\frac{d  f_\G{n}}{dt} =-\ad^{\ast}_{\frac{\partial \C{H}}{\p f_\G{n}}} f_\G{n} - \frac{\partial \C{H}}{\partial f_\G{s}} \overset{\ast }{\rt}{f_\G{n}} - 
\mathfrak{b}_{\frac{\partial \C{H}}{\partial f_\G{s}}}^{\ast}f_\G{s}.
\end{align*}
The claim then follows from Corollary \ref{coaddecompAexp-FQpp} substituting
\begin{align*}
&\hat{\s} \to \frac{\p \C{H}}{\p f_\G{s}}, \qquad \s \to \frac{\p \C{H}}{\p f_{(0)}}, \qquad Y \to \frac{\p \C{H}}{\p f_{(1)}},\\ 
&\hat{\B{X}} \to \frac{\p \C{H}}{\p f_\G{n}}, \qquad \hat{\B{X}}^k \to \frac{\p \C{H}}{\p f_{(k)}}.
\end{align*}
\end{proof}

\subsubsection*{Euler's fuid as a sub-dynamics of the Vlasov plasma} 

Giving up the $\hat{\G{n}}^\ast$-dependence of the Hamiltonian, we arrive at the Euler's fluid motion in \ref{momEuler} as a subdynamics on $\C{F}^\ast_0(T^*\C{Q})$ as \begin{equation} \label{coad-result-V-0}
\frac{d f_\G{s}}{dt}=\{  \frac{\partial \C{H}}{\partial f_\G{s}},f_{(0)}\} + \{ \frac{\partial \C{H}}{\partial f_{(1)}},f_{(1)}\}.
\end{equation}

\subsubsection*{Kinetic moments of order $\geq 2$ as sub-dynamics of the Vlasov plasma} 

This time assuming the Hamiltonian to be depended only on $\hat{\G{n}}^\ast$, we arrive at 
\begin{equation} \label{coad-result-V-1}
\frac{d f_\G{n}}{dt}= \sum_{k\geq 2}\{\frac{\partial \C{H}}{\partial f_{(\ell+k)}},f_{(\ell)}\}.
\end{equation}

\subsubsection*{Decomposition of the Vlasov equation} 

In particular, for the Hamiltonian functional given by
\[
\C{H}(f):=\int_{T^\ast\C{Q}}\,f(q,p)h(q,p)\,dqdp
\]
where $h\in \C{F}_0(T^\ast\C{Q})$ is the one in \eqref{total-energy}, we have
\[
\frac{\p \C{H}}{\p f} = h.
\] 
As such, the double cross sum decomposition of the Vlasov equation \eqref{Vlasov-pre} may be obtained by 
\[
\frac{\p \C{H}}{\p f_\G{s}} = \frac{\p \C{H}}{\p f_{(0)}} \to e\phi, \qquad \frac{\p \C{H}}{\p f_\G{n}} = \frac{\p \C{H}}{\p f_{(2)}} \to \frac{1}{2m}p^2, \qquad \frac{\p \C{H}}{\p f_{(1)}}\to 0, \qquad \frac{\p \C{H}}{\p f_{(k)}} \to 0, \quad k\geq 3,
\]
in \eqref{coad-mp-TQpp}.

\section{Momentum-Vlasov dynamics}\label{mp-mVla-sec}

In this section  we study the Lie-Poisson theory through the Lie algebra of Hamiltonian vector fields. To this end, in accordance with the previous section, we introduce the graded subalgebra of the non-flat Hamiltonian vector fields. Then, also parallel to the previous sections, we consider the dual space of the Hamiltonian vector fields, and present the coadjoint action. We conclude the section with the presentation of the matched Lie-Poisson equations.

\subsection{The Lie algebra of Hamiltonian vector fields}\label{subsect-Ham-Lie-alg}~

The cotangent bundle $T^\ast\C{Q}$ is an exact symplectic manifold by admitting the canonical (Liouville) 1-from $\theta_{\C{Q}}:=p_\ell dq^\ell$, and the symplectic 2-form  $\om_{\C{Q}}=-d\theta_{\C{Q}} = dq^\ell \wedge dp_\ell$. 
A vector field $X\in \G{X}(T^*\C{Q})$ is called ``symplectic'' if its interior product $\iota_X\omega_{\C{Q}} $ with the symplectic 2-form is closed, and $X\in \G{X}(T^*\C{Q})$ is called a (globally) ``Hamiltonian'' if its interior product with the symplectic 2-form is exact; more precisely, 
\begin{equation} \label{HamEq}
\iota_{X} \omega_{\C{Q}}=dh
\end{equation}
for some $h\in \C{F}(T^*\C{Q})$, \cite{JansVism16,MarsdenRatiu-book}. 
In the latter case $X\in \G{X}(T^*\C{Q})$ is said to be the Hamiltonian vector field of $h\in \C{F}(T^*\C{Q})$, and is denoted by $X_h$. The space $\G{X}_{\rm Ham}(T^\ast\C{Q})$ of Hamiltonian vector fields is a Lie algebra through the opposite Jacobi-Lie bracket of vector fields. Moreover,
\begin{equation}\label{epi-onto-Ham}
\vp:\C{F}(T^\ast\C{Q}) \to \G{X}_{\rm Ham}(T^\ast\C{Q}), \qquad h \mapsto -X_h:=-\frac{\p h}{\p p_\ell}\frac{\p }{\p q^\ell} + \frac{\p h}{\p q^\ell}\frac {\p }{\p p_\ell}
\end{equation}
is a Lie algebra epimorphism, whose kernel consists of the constant functions. 

Our choice of the opposite Jacobi-Lie bracket for the Hamiltonian vector fields, and the opposite Poisson bracket for the smooth functions, is motivated by the symmetry of the plasma theory; namely, the particle relabelling symmetry being given by a right action, \cite{MaWe81}.

It worths to note that the composition of the maps \eqref{TQ-to-F} and \eqref{epi-onto-Ham} generalizes (the negative of) the ``complete cotangent lift'' of vector fields, see for instance \cite{LeRo,MarsdenRatiu-book,YaPa67}, and will be referred here as the ``generalized complete cotangent lift'' (GCCL in short) of tensor fields. More precisely,
\begin{align}\label{GCCL}
\begin{split}
& \text{{\small GCCL}}: \G{T}\C{Q}\to \G{X}_{\rm Ham,0}(T^*\C{Q}), \\
&\text{{\small GCCL}} ( \mathbb{X}^k )
:= -k\mathbb{X}^{i_{1}\ldots i_{k-1}\ell}(q)p_{i_{1}}p_{i_{2}}\ldots p_{i_{k-1}}\frac{\p }{\p q^{\ell}} +\mathbb{X}^{i_{1}i_{2}\ldots i_{k}}_{,\ell}(q) p_{i_{1}}p_{i_{2}}\ldots p_{i_{k}}\frac{\p}{\p p_\ell} 
\end{split}
\end{align}
see also \cite{EsGrGuPa19,No93}. Being a composition of two Lie algebra homomorphisms, \eqref{GCCL} is clearly a Lie algebra homomorphism.

It now readily follows from \eqref{GTQ-matched-pair}, \eqref{FTQ-isom-FQp}, and \eqref{TQ-to-F} that
\begin{equation}\label{Ham-vf}
\G{X}_{\rm Ham}(T^\ast\C{Q}) \cong \C{F}(T^\ast\C{Q}) / \B{R}\cong m^\infty_{\C{Q}\times \{0\}}\rtimes \Big((\hat{\G{s}} / \B{R} )\bowtie \hat{\G{n}}\Big).
\end{equation}
We shal then call the Hamiltonian vector fields
\begin{equation}\label{formal-Ham-vf}
\G{X}_{\rm Ham,0}(T^*\C{Q}):=(\hat{\G{s}} / \B{R} )\bowtie \hat{\G{n}}=: {\G{s}^c} \bowtie {\G{n}^c} 
\end{equation}
corresponding to the non-flat functions to be the ``non-flat Hamiltonian vector fields''. 

Let us conclude the present subsection with the explicit expressions of the mutual actions of ${\G{s}^c}$ and ${\G{n}^c}$. Given any $\hat{\s} \in \hat{\G{s}}$, and any $\hat{\B{X}} \in \hat{\G{n}}$, we have
\[
X_{\hat{\B{X}}}\rt X_{\hat{\s}} + X_{\hat{\B{X}}}\lt X_{\hat{\s}} = -[X_{\hat{\B{X}}},X_{\hat{\s}}] = X_{\{\hat{\B{X}},\hat{\s}\}} = X_{\{\hat{\B{X}}^2,\s\}} + \Big(X_{\{\hat{\B{X}},Y\}} + \sum_{k\geq 3} X_{\{\hat{\B{X}}^k,\s\}}\Big) \in {\G{s}^c}\bowtie {\G{n}^c}.
\]
Accordingly,
\begin{align}
& X_{\hat{\B{X}}}\rt X_{\hat{\s}} = X_{\{\hat{\B{X}}^2,\s\}} =  -2\s_{,\ell}(q)\B{X}^{j\ell}(q)\frac{\p }{\p q^j} + 2\Big(\s_{,\ell j}(q)\B{X}^{i\ell}(q) + \s_{,\ell}(q)\B{X}^{i\ell}_{,j}(q)\Big)p_i\frac{\p}{\p p_j},  \label{XG-rt-XF-Ham}\\
& X_{\hat{\B{X}}}\lt X_{\hat{\s}} = X_{\{\hat{\B{X}},Y\}} + \sum_{k\geq 3} X_{\{\hat{\B{X}}^k,\s\}} = \label{XG-lt-XF-Ham}\\
& \sum_{k\geq 2} \Big\{k\Big(-kY^j_{,\ell}(q)\B{X}^{i_1\ldots i_{k-1}\ell}(q) +Y^\ell(q) \B{X}^{i_1\ldots i_{k-1}j}_{,\ell}\Big)p_{i_1}\ldots p_{i_{k-1}}\frac{\p}{\p q^j} + \notag\\
& \Big(-k \B{X}^{i_1\ldots i_{k-1}\ell}(q)Y^{i_k}_{,\ell j} - k\B{X}^{i_1\ldots i_{k-1}\ell}_{,j}(q)Y^{i_k}_{,\ell }(q) +Y^\ell_{,j}(q)\B{X}^{i_1\ldots i_k}_{,\ell}(q) + Y^\ell(q)\B{X}^{i_1\ldots i_k}_{,\ell j}(q)\Big) p_{i_1}\ldots p_{i_k}\frac{\p}{\p p_j}\Big\} +\notag\\
& \sum_{k\geq 3} \Big\{-k(k-1)\s_{,\ell}(q)\B{X}^{i_1\ldots i_{k-2}j\ell}(q)p_{i_1}\ldots p_{i_{k-2}}\frac{\p}{\p q^j}  + \notag\\
& \hspace{4cm}\Big(k\s_{,\ell j}(q)\B{X}^{i_1\ldots i_{k-1}\ell}(q) + k\s_{,\ell}(q)\B{X}^{i_1\ldots i_{k-1}\ell}_{,j}(q)\Big)p_{i_1}\ldots p_{i_{k-1}}\frac{\p}{\p p_j} \Big\}.   \notag
\end{align}

\subsection{The space of 1-forms with non-trivial divergences}\label{subsect-1-form-div}~

Along the lines of the previous paragraphs, we begin with the following characterization of the dual space of the Lie algebra of Hamiltonian vector fields, see also \cite{Gu10}.

\begin{proposition} \label{dualHam}
The dual space $\G{X}_{{\rm Ham}}^{\ast }( T^{\ast }\mathcal{Q})$ of the Lie algebra $\G{X}_{{\rm Ham}}( T^{\ast }\mathcal{Q}) $ of Hamiltonian vector fields may be given by 
\begin{equation} \label{momdef}
\G{X}_{{\rm Ham}}^{\ast }( T^{\ast }\mathcal{Q}) =\{\Pi \in\Lambda^{1}(T^{\ast }\mathcal{Q})\mid {\rm div}\Pi^{\sharp}\neq  0 \text{ if } \Pi \neq 0\},
\end{equation}
and,  
\[
\omega _{\mathcal{Q}}^{\sharp
}:\Lambda^1(T^*\C{Q}) \to \mathfrak{X}(T^*\C{Q}), \qquad \Pi_\ell dq^\ell + \Pi^\ell dp_\ell=:\Pi \mapsto \Pi^\sharp :=\Pi^\ell \frac{\p}{\p q^\ell} - \Pi_\ell \frac{\p}{\p p_\ell}
\]
is the musical isomorphism induced from the symplectic 2-form $\omega _{\mathcal{Q}}:=dq^\ell \wedge dp_\ell \in \Lambda^2(T^\ast\C{Q})$.
\end{proposition} 

\begin{proof}
We have
\begin{align} \label{dual-calc}
\begin{split}
&\langle X_h,\Pi \rangle = \int_{T^{\ast }\C{Q}}\,\langle X_{h} ,\Pi\rangle \,dqdp
=-\int_{T^{\ast }\C{Q}}\,\langle dh,\Pi^{\sharp
}\rangle\,dqdp =  -\int_{T^{\ast }\mathcal{Q}} \, \Pi ^{\sharp
}(h) \, dqdp = \\
&  -\int_{T^{\ast }\mathcal{Q}} \, \left(\Pi^\ell\frac{\p h}{\p q^\ell} - \Pi_\ell\frac{\p h}{\p p_\ell}\right) \, dqdp = \int_{T^{\ast }\mathcal{Q}} \, \left(\frac{\p \Pi^\ell}{\p q^\ell} -\frac{\p \Pi_\ell}{\p p_\ell}\right) h\, dqdp  =  \int_{T^{\ast }\mathcal{Q}} \, {\rm div}(\Pi^\sharp) h\,dqdp,
\end{split}
\end{align}
where we use the integration by parts on the fifth equality. On the other hand, \eqref{epi-onto-Ham} having the kernel consisting of the constant functions, \eqref{dual-calc} is well-defined as it vanishes whenever $h \in \C{F}(T^\ast\C{Q})$ is a constant function. Finally, the non-degeneracy of \eqref{dual-calc} follows from the definition; namely, the divergence being non-zero for non-zero 1-form densities.
\end{proof}

Accordingly, \eqref{epi-onto-Ham} transposes into the injection
\[
\G{X}_{{\rm Ham}}^{\ast }( T^{\ast }\mathcal{Q}) \to \C{F}^\ast(T^\ast\C{Q}) \cong \C{F}(T^\ast\C{Q}), \qquad \Pi\mapsto {\rm div}(\Pi^\sharp).
\]
In particular,  
\begin{equation}\label{Ham-to-FTQ}
\G{X}^\ast_{\rm Ham,0}(T^*\C{Q}) \to \C{F}^\ast_0(T^*\C{Q}), \qquad \Pi\mapsto {\rm div}(\Pi^\sharp).
\end{equation}
yields the decompostion 
\begin{equation}\label{X_Ham-dual-decomp}
\G{X}^\ast_{\rm Ham,0}(T^*\C{Q}) \cong (\G{s}^c)^\ast \bowtie (\G{n}^c)^\ast,
\end{equation}
where
\begin{align*}
& (\G{s}^c)^\ast := \{\Pi\in \Lambda^1(T^\ast\C{Q})\mid \int_{T^\ast\C{Q}}\,p^k {\rm div}(\Pi^\sharp)\,dp = 0, \quad \forall\,k\geq 2\}, \\ 
&(\G{n}^c)^\ast := \{\Pi\in \Lambda^1(T^\ast\C{Q})\mid \int_{T^\ast\C{Q}}\,{\rm div}(\Pi^\sharp)\,\,dp = \int_{T^\ast\C{Q}}\,p {\rm div}(\Pi^\sharp)\,\,dp =0\}.
\end{align*}
More precisely, parallel to \eqref{h-decomp}, we have the moment decomposition
\[
\G{X}_{\rm Ham,0}(T^*\C{Q})\ni \Pi \mapsto \Pi^\G{s} + \Pi^\G{n} := \Big(\Pi^{(0)} + \Pi^{(1)}\Big) + \sum_{k\geq 2}  \Pi^{(k)} \in (\G{s}^c)^\ast\oplus (\G{n}^c)^\ast,
\]
through the injection \eqref{Ham-to-FTQ}.

Just like the previous paragraphs, once again we conclude with the computation of the coadjoint action. Along the lines of of \cite{Gu10}, see also \cite{esen2012geometry}, the coadjoint action may also be represented in terms of Hamiltonian operators.

\begin{proposition}\label{prop-coad-III}
The coadjoint action of $X_h\in \G{X}_{\rm Ham}(T^\ast\C{Q})$ on $\Pi=\Pi_\ell dq^\ell + \Pi^\ell dp_\ell\in \G{X}^\ast_{\rm Ham}(T^\ast\C{Q})$ is given by
\begin{equation}\label{Coad-VF-Ham}
\ad^\ast_{X_h}\Pi = J_{LP}(\Pi)(X_h),
\end{equation}
where $\om_{\mathcal{Q}}:=dq^\ell \wedge dp_\ell \in \Lambda^2(T^\ast\C{Q})$ is the symplectic 2-form, and
\[
J_{LP}(\Pi) := -\left(\begin{array}{ll} 
\Pi_i\frac{\p}{\p q^j} + \frac{\p \Pi_j}{\p q^i} + \Pi_j\frac{\p}{\p q^i} & \Pi^i\frac{\p}{\p q^j} + \frac{\p \Pi_j}{\p p_i} + \Pi_j\frac{\p}{\p p_i} \\
\Pi_i\frac{\p}{\p p_j} + \frac{\p \Pi^j}{\p q^i} + \Pi^j\frac{\p}{\p q^i} & \Pi^i\frac{\p}{\p p_j} + \frac{\p \Pi^j}{\p p_i} + \Pi^j\frac{\p}{\p p_i}
\end{array}\right).
\]
\end{proposition}

\begin{proof}
Given any $X_g \in \G{X}_{\rm Ham}(T^\ast\C{Q})$, we see at once that
\begin{align*}
& \langle \ad^\ast_{X_h}\Pi, X_g\rangle = \langle \Pi, -[X_g,X_h]\rangle = \langle \Pi, X_{\{g,h\}}\rangle = \\
& \int_{T^\ast\C{Q}}\,{\rm div}(\Pi^\sharp)\{g,h\}\,dqdp =  \int_{T^\ast\C{Q}}\,-\{{\rm div}(\Pi^\sharp),h\}g\,dqdp = \\ 
& \int_{T^\ast\C{Q}}\,\left[-\frac{\p }{\p q^s}\left(\frac{\p \Pi^\ell}{\p q^\ell}-\frac{\p \Pi_\ell}{\p p_\ell}\right)\frac{\p h}{\p p_s} + \frac{\p h}{\p q^s}\frac{\p }{\p p_s}\left(\frac{\p \Pi^\ell}{\p q^\ell}-\frac{\p \Pi_\ell}{\p p_\ell}\right)\right]g\,dqdp = \\
& \int_{T^\ast\C{Q}}\,\left[-\frac{\p }{\p q^\ell}\left(\frac{\p \Pi^\ell}{\p q^s}\frac{\p h}{\p p_s} - \frac{\p \Pi^\ell}{\p p_s}\frac{\p h}{\p q^s}\right) + \frac{\p }{\p p_\ell}\left(\frac{\p \Pi_\ell}{\p q^s}\frac{\p h}{\p p_s} - \frac{\p \Pi_\ell}{\p p_s}\frac{\p h}{\p q^s}\right)\right]g\,dqdp  +\\
&  \int_{T^\ast\C{Q}}\, \left[\frac{\p \Pi^\ell}{\p q^s}\frac{\p^2 h}{\p q^\ell \p p_s} - \frac{\p \Pi^\ell}{\p p_s}\frac{\p^2 h}{\p q^\ell \p q^s} - \frac{\p \Pi_\ell}{\p q^s}\frac{\p^2 h}{\p p_\ell \p p_s} + \frac{\p \Pi_\ell}{\p p_s}\frac{\p^2 h}{\p p_\ell \p q^s}\right] g\,dqdp =\\
& \int_{T^\ast\C{Q}}\,-\left[\left(\frac{\p \Pi_\ell}{\p q^s}\frac{\p h}{\p p_s} - \frac{\p \Pi_\ell}{\p p_s}\frac{\p h}{\p q^s}\right)\frac{\p g}{\p p_\ell}-\left(\frac{\p \Pi^\ell}{\p q^s}\frac{\p h}{\p p_s} - \frac{\p \Pi^\ell}{\p p_s}\frac{\p h}{\p q^s}\right)\frac{\p g}{\p q^\ell} \right]\,dqdp + \\
& \hspace{2cm} -\int_{T^\ast\C{Q}}\, \left(\left[\Pi_\ell\frac{\p^2 h}{\p p_\ell \p q^s} - \Pi^\ell\frac{\p^2 h}{\p q^\ell \p q^s} \right]\frac{\p g}{\p p_s} -  \left[\Pi_\ell\frac{\p^2 h}{\p p_\ell \p p_s} - \Pi^\ell \frac{\p^2 h}{\p q^\ell \p p_s}\right] \frac{\p g}{\p q^s}\right)\,dqdp.
\end{align*}
The claim, then, follows.
\end{proof}

In particular, for \eqref{total-energy}, the Lie-Poisson equation 
\[
\frac{d \Pi}{dt}=-\ad^*_{-X_h} \Pi
\]
yields the momentum-Vlasov equations
\begin{align}\label{momvla}
\begin{split}
&\frac{d\Pi _{i}}{dt} =-{X_{h}}(\Pi _{i}) +e\frac{\partial ^{2}\phi }{\partial q^{i}\partial q^{j}}\Pi ^{j}  \\
&\frac{d\Pi ^{i}}{dt} =-{X_{h}}(\Pi ^{i})-\frac{1}{m}\delta ^{ij}\Pi _{j},
\end{split}
\end{align} 
see also \cite{esen2012geometry,Gu10}.

\subsection{Momentum-Vlasov dynamics}\label{Sec-m-Vlasov}~

\begin{proposition} \label{prop-dualaction-Ham}
The left action \eqref{XG-rt-XF-Ham} gives rise to the right action
\begin{equation} \label{dualaction-I-FQpp}
\overset{\ast }{\lt}:{(\G{s}^c)}^*\otimes {\G{n}^c}\to {(\G{s}^c)}^*, \qquad \Pi^\G{s}\overset{\ast }{\lt}X_{\hat{\B{X}}}:= -J_{LP}(\Pi^\G{s})(X_{\hat{\B{X}}})
\end{equation}  
for any $\hat{\B{X}}\in \hat{\G{n}}$, and any $\Pi^\G{s} \in (\G{s}^c)^\ast$. Similarly, given any $\hat{\s}\in \hat{\G{s}}$, and $\Pi^\G{n} \in (\G{n}^c)^\ast$, 
the right action \eqref{XG-lt-XF-Ham} yields the left action 
\begin{equation} \label{dualaction-II-FQpp}
\overset{\ast }{\rt}:{\G{s}^c}\otimes (\G{n}^c)^\ast\to{(\G{n}^c)}^*, \qquad X_{\hat{\s}}\overset{\ast }{\rt} \Pi^\G{n} = J_{LP}(\Pi^\G{n})(X_{\hat{\s}}).
\end{equation}
\end{proposition}

\begin{proof}
Along the lines of Proposition \ref{prop-coad-III}, we see at once that
\[
\langle\Pi^\G{s}\overset{\ast }{\lt}X_{\hat{\B{X}}}, X_{\hat{\s}}\rangle = \langle \Pi^\G{s}, X_{\hat{\B{X}}}\rt X_{\hat{\s}}\rangle = \langle \Pi^\G{s}, X_{\{\hat{\B{X}},\hat{\s}\}}\rangle = \langle -J_{LP}(\Pi^\G{s})(X_{\hat{\B{X}}}),X_{\hat{\s}}\rangle,
\]
and that
\[
\langle X_{\hat{\s}}\overset{\ast }{\rt}\Pi^\G{n}, X_{\hat{\B{X}}}\rangle = \langle \Pi^\G{n}, X_{\hat{\B{X}}}\lt X_{\hat{\s}}\rangle = \langle \Pi^\G{n}, X_{\{\hat{\B{X}},\hat{\s}\}}\rangle = \langle J_{LP}(\Pi^\G{n})(X_{\hat{\s}}),X_{\hat{\B{X}}}\rangle.
\]
Both claims thus follow.
\end{proof}

We shall now investigate the transposes of the mappings 
\begin{align} 
&\G{b}_{X_{\hat{\s}}} :{\G{n}^c} \to {\G{s}^c}, \quad \G{b}_{X_{\hat{\s}}}X_{\hat{\B{X}}} = X_{{\hat{\B{X}}} }\rt X_{{\hat{\s}}}, \qquad\G{a}_{X_{\hat{\B{X}}}}:{\G{s}^c} \to {\G{n}^c}, \quad \G{a}_{X_{\hat{\B{X}}}}X_{\hat{\s}} = X_{{\hat{\B{X}}} }\lt X_{ {\hat{\s}}}, \label{a-b-ex-Ham} \\
& \ad_{\hat{\s}}:{\G{s}^c}\to {\G{s}^c}, \quad \ad_{X_{\hat{\s}}}X_{\hat{\lambda}} = X_{\{{\hat{\s}},{\hat{\lambda}}\}}, \qquad \ad_{\hat{\B{X}}}:{\G{n}^c}\to {\G{n}^c}, \quad \ad_{X_{\hat{\B{X}}}}X_{\hat{\B{Y}}}= X_{\{{\hat{\B{X}}},{\hat{\B{Y}}}\}},\label{ad-Ham}
\end{align}
where ${\hat{\s}},{\hat{\lambda}}\in \hat{\G{s}}$, and ${\hat{\B{X}}},{\hat{\B{Y}}}\in \hat{\G{n}}$, keeping this time Proposition \ref{prop-coad-III}, in addition to \eqref{coad}, in mind. Just as above, we shall proceed along
\[
J_{LP}(\Pi^\G{s})(X_{\hat{\s}}) = \ad^\ast_{X_{\hat{\s}}}\Pi^\G{s} + \G{b}^\ast_{X_{\hat{\s}}}\Pi^\G{s} \in (\G{s}^c)^\ast \oplus (\G{n}^c)^\ast, \qquad J_{LP}(\Pi^\G{n})(X_{\hat{\B{X}}}) = -\G{a}^\ast_{X_{\hat{\B{X}}}}\Pi^\G{n} + \ad^\ast_{X_{\hat{\B{X}}}}\Pi^\G{n} \in (\G{s}^c)^\ast \oplus (\G{n}^c)^\ast
\]
for any ${\hat{\s}}\in \hat{\G{s}}$, any ${\hat{\B{X}}}\in \hat{\G{n}}$, any $\Pi^\G{s}\in (\G{s}^c)^\ast$, and any $\Pi^\G{n}\in (\G{n}^c)^\ast$.

\begin{proposition} \label{cross-act-FQpp}
Given $\hat{\s}\in \hat{\G{s}}$ with $\hat{\B{X}} \in \hat{\G{n}}$, and $\Pi^\G{s}\in (\G{s}^c)^\ast$ with $\Pi^\G{n}\in (\G{n}^c)^\ast$, the transposes of the linear operators in \eqref{a-b-ex-Ham} and \eqref{ad-Ham} may be given by
\begin{align}
 & \ad^\ast_{X_{\hat{\s}}} : (\G{s}^c)^\ast \to (\G{s}^c)^\ast,\quad \ad^\ast_{X_{\hat{\s}}}\Pi^\G{s} = J_{LP}(\Pi^{(0)})(X_{\hat{\s}}) + J_{LP}(\Pi^{(1)})(X_{Y}), \label{ad-ex-dual-Ham}\\
 &\G{b}^\ast_{X_{\hat{\s}}}:(\G{s}^c)^\ast \to (\G{n}^c)^\ast, \quad \G{b}^\ast_{X_{\hat{\s}}}\Pi^\G{s} =  J_{LP}(\Pi^{(1)})(X_{\s}),
 \label{b-ex-dual-Ham} \\
&\ad^\ast_{X_{\hat{\B{X}}}} : (\G{n}^c)^\ast \to (\G{n}^c)^\ast,\quad \ad^\ast_{X_{\hat{\B{X}}}}\Pi^\G{n} = \sum_{k\geq 2} J_{LP}(\Pi^{(\ell)})(X_{{\hat{\B{X}}}^{\ell+k}}) , \label{ad-ex-dual-Ham}\\
&\G{a}^\ast_{X_{\hat{\B{X}}}} :(\G{n}^c)^\ast \to (\G{s}^c)^\ast, \quad \G{a}^\ast_{X_{\hat{\B{X}}}}\Pi^\G{n}= - J_{LP}(\Pi^{(k)})(X_{{\hat{\B{X}}}^k}) -J_{LP}(\Pi^{(k)})(X_{{\hat{\B{X}}}^{k+1}}). \label{a-ex-dual-Ham}  
\end{align}
\end{proposition}

\begin{proof}
Along the lines of Proposition \ref{prop-coad-III} we have
\[
\langle J_{LP}(\Pi^\G{s})(X_{\hat{\s}}), X_h\rangle = \int_{T^\ast\C{Q}}\,-\{{\rm div}({\Pi^\G{s}}^\sharp),\hat{\s}\}h \,dqdp
\]
for any $h\in \C{F}_0(T^*\C{Q})$. On the other hand, following \eqref{F-f} we have
\[
\{{\rm div}({\Pi^\G{s}}^\sharp),\hat{\s}\} = \Big(\{{\rm div}({\Pi^\G{s}}^\sharp)_{(0)},\hat{\s}\} + \{{\rm div}({\Pi^\G{s}}^\sharp)_{(1)}, Y\}\Big) + \{{\rm div}({\Pi^\G{s}}^\sharp)_{(1)},\s\} \in \hat{\G{s}}^\ast \oplus \hat{\G{n}}^\ast.
\]
Accordingly,
\[
J_{LP}({\Pi^\G{s}})(X_{\hat{\s}}) =  \Big(J_{LP}(\Pi^{(0)})(X_{\hat{\s}}) + J_{LP}(\Pi^{(1)})(X_{Y}) \Big) + J_{LP}(\Pi^{(1)})(X_{\s})  \in (\G{s}^c)^\ast \oplus (\G{n}^c)^\ast.
\]
Similarly, it follows from \eqref{G-g} that
\[
 J_{LP}({\Pi^\G{n}})(X_{\hat{\B{X}}})= \Big( J_{LP}(\Pi^{(k)})(X_{{\hat{\B{X}}}^k}) +J_{LP}(\Pi^{(k)})(X_{{\hat{\B{X}}}^{k+1}})\Big)+ \sum_{k\geq 2} J_{LP}(\Pi^{(\ell)})(X_{{\hat{\B{X}}}^{\ell+k}}) \in (\G{s}^c)^\ast \oplus (\G{n}^c)^\ast.
\]
\end{proof}

The double cross sum decomposition of the coadjoint action of $\G{X}_{\rm Ham,0}(T^*\C{Q})$ on $\G{X}^\ast_{\rm Ham,0}(T^*\C{Q})$ is now a direct consequence.

\begin{corollary} \label{coaddecompAexp-Ham}
Given any $(X_{\hat{\s}},X_{\hat{\B{X}}})\in {\G{s}^c}\bowtie {\G{n}^c}$ where $\hat{\s}=\s+Y$ and $\hat{\B{X}}=\sum_{k\geq2}\hat{\B{X}}^k$, and any $(\Pi^\G{s},\Pi^\G{n})\in (\G{s}^c)^\ast\oplus (\G{n}^c)^\ast$, the coadjoint action of $\G{X}_{\rm Ham,0}(T^*\C{Q})$ on $\G{X}^\ast_{\rm Ham,0}(T^*\C{Q})$ may be given by
\begin{align}\label{coad-Ham-decomposed}
\begin{split}
& \ad^\ast_{(X_{\hat{\s}}+ X_{\hat{\B{X}}})}(\Pi^\G{s} +\Pi^\G{n}) = \\
& \Big(J_{LP}(\Pi^{(0)})(X_{\hat{\s}}) + J_{LP}(\Pi^{(1)})(X_{Y}) - J_{LP}(\Pi^\G{s})(X_{\hat{\B{X}}}) +  J_{LP}(\Pi^{(k)})(X_{{\hat{\B{X}}}^k}) +J_{LP}(\Pi^{(k)})(X_{{\hat{\B{X}}}^{k+1}})\Big) + \\
&\hspace{3cm} \Big(\sum_{k\geq 2} J_{LP}(\Pi^{(\ell)})(X_{{\hat{\B{X}}}^{\ell+k}}) - J_{LP}(\Pi^\G{n})(X_{\hat{\s}})+J_{LP}(\Pi^{(1)})(X_{\s})\Big) \in \hat{\G{s}}^\ast\oplus \hat{\G{n}}^\ast.
\end{split}
\end{align}
\end{corollary}

We thus conclude the momentum-Vlasov equations below.

\begin{corollary}\label{coroll-matched-mom-Vlasov}
The Lie-Poisson equations on $\G{X}^\ast_{\rm Ham,0}(T^*\C{Q})=(\G{s}^c)^\ast \oplus (\G{n}^c)^\ast$, generated by a Hamiltonian functional $\C{H}=\C{H}(\Pi^\G{s},\Pi^\G{n})$, may be given by
\begin{equation} \label{coad-mp-Ham}
\begin{split}
& \frac{d \Pi^\G{s}}{dt}=-J_{LP}(\Pi^{(0)})\left(\frac{\p \C{H}}{\p \Pi^\G{s}}\right) - J_{LP}(\Pi^{(1)})\left(\frac{\p \C{H}}{\p \Pi^{(1)}}\right) +\\
& \hspace{4cm}  J_{LP}(\Pi^\G{s})\left(\frac{\p \C{H}}{\p \Pi^\G{n}}\right) - J_{LP}(\Pi^{(k)})\left(\frac{\p \C{H}}{\p \Pi^{(k)}}\right) - J_{LP}(\Pi^{(k)})\left(\frac{\p \C{H}}{\p \Pi^{k+1}}\right) ,
\\
& \frac{d \Pi^\G{n}}{dt}= \sum_{k\geq 2} J_{LP}(\Pi^{(\ell)})\left(\frac{\p \C{H}}{\p \Pi^{(\ell+k)}}\right) + J_{LP}(\Pi^\G{n})\left(\frac{\p \C{H}}{\p \Pi^\G{s}}\right) -  J_{LP}(\Pi^{(1)})\left(\frac{\p \C{H}}{\p \Pi^{(0)}}\right).
\end{split}
\end{equation}
\end{corollary}

\begin{proof}
It follows from \eqref{LPEgh} that
\begin{align*}
&\frac{d \Pi^\G{s}}{dt}=-\ad^{\ast}_{\frac{\partial \C{H}}{\partial \Pi^\G{s}}} \Pi^\G{s} +
\Pi^\G{s} \overset{\ast }{\lt}\frac{\partial \C{H}}{\partial \Pi^\G{n}} + 
\mathfrak{a}_{\frac{\partial \C{H}}{\p \Pi^\G{n}}}^{\ast}{ \Pi^\G{n}},
\\
&\frac{d  \Pi^\G{n}}{dt} =-\ad^{\ast}_{\frac{\partial \C{H}}{\p \Pi^\G{n}}}\Pi^\G{n} -
\frac{\partial \C{H}}{\partial \Pi^\G{s}} \overset{\ast }{\rt}{\Pi^\G{n}} - 
\mathfrak{b}_{\frac{\partial \C{H}}{\partial \Pi^\G{s}}}^{\ast}\Pi^\G{s}.
\end{align*}
The claim then follows from Corollary \ref{coaddecompAexp-FQpp} substituting
\begin{align*}
& X_{\hat{\s}} \to \frac{\p \C{H}}{\p \Pi^\G{s}}, \qquad X_{\s} \to \frac{\p \C{H}}{\p \Pi^{(0)}}, \qquad X_{Y} \to \frac{\p \C{H}}{\p \Pi^{(1)}}, \\
& X_{\hat{\B{X}}} \to \frac{\p \C{H}}{\p \Pi^\G{n}}, \qquad X_{\hat{\B{X}}^k} \to \frac{\p \C{H}}{\p \Pi^{(k)}}.
\end{align*}
\end{proof}

\subsubsection*{The decomposition of the momentum-Vlasov equations} 

Along the lines of \cite{Gu10}, the Hamiltonian functional
\[
\C{H}(\Pi):= \int_{T^\ast\C{Q}}\,\langle \Pi,-X_h\rangle\,dqdp,
\]
where $-X_h \in \G{X}_{\rm Ham,0}(T^\ast\C{Q})$ being the Hamiltonian vector field corresponding to the function \eqref{total-energy}, satisfies
\[
\frac{\p \C{H}}{\p \Pi} = -X_h.
\] 
Accordingly, the double cross sum decomposition of the momentum-Vlasov equations \eqref{momvla} is achieved by substituting 
\[
\frac{\p \C{H}}{\p \Pi^\G{s}} = \frac{\p \C{H}}{\p \Pi^{(0)}} \to -eX_\phi, \qquad \frac{\p \C{H}}{\p \Pi^\G{n}} = \frac{\p \C{H}}{\p \Pi^{(2)}} \to -\frac{1}{2m}X_{p^2}, \qquad \frac{\p \C{H}}{\p \Pi^{(1)}}\to 0, \qquad \frac{\p \C{H}}{\p \Pi^{(k)}} \to 0, \quad k\geq 3,
\]
in \eqref{coad-mp-Ham}.

\section{The group of canonical diffeomorphisms}\label{sect-can-diff}

In this final section, we shall study the Lie group counterpart of the decompositions we have considered in the previous sections. To this end, we begin with a quick overview of the matched pairs of Lie groups, and their double cross products. Then, we shall present the double cross product group structure of the group of canonical diffeomorphisms of the cotangent bundle, preserving the canonical 1-form, the Lie algebra of which may be identified with the Lie algebra of non-flat Hamiltonian vector fields.

\subsection{Matched pairs of Lie groups}\label{subsect-matched-Lie-gr}~

Parallel to the Lie algebra case, we shall now recall briefly the matched pair theory (and hence the double cross product construction) for Lie groups from \cite{Maji90,Majid-book,LuWein90,Maji90-II,Ta81,Zhan10}. 

Let $(G, H)$ be a pair of Lie groups, with mutual actions
\begin{align} 
& \rt:H\times G\to G,\quad (y,x) \mapsto y\rt x, \label{Lieact-left-gr}\\ 
&\lt:H\times G\to H, \quad (y,x) \mapsto y\lt x \label{Lieact-right-gr}.
\end{align}
The pair $(G, H)$, then, is called a ``matched pair of Lie groups'' if the mutual actions \eqref{Lieact-left-gr} and \eqref{Lieact-right-gr} satisfy
\begin{equation} \label{compcon-mpl-gr}
\begin{split}
y\rt (x_1x_2)=(y\rt x_1)((y\lt x_1)\rt x_2), \\
(y_1y_2)\lt x = (y_1\lt(y_2\rt x))(y_2\lt x),
\end{split}
\end{equation}
for any $x,x_1,x_2\in G$, and any $y,y_1,y_2\in H$. Now, given a matched pair of Lie groups $(G,H)$, the product space $G\bowtie H:= G\times H$ becomes a Lie group through
\begin{equation}\label{mpla-gr}
(x_1,y_1)(x_2,y_2)=\Big(x_1(y_1\rt x_2) , (y_1\lt x_2)y_2\Big)
\end{equation}
for any $(x_1,y_1),(x_2,y_2)\in G\bowtie H$, called the ``double cross product'' of $G$ and $H$.

Just as in the Lie algebra case, the double cross product group $G\bowtie H$ reduces to the (right-handed) semi-direct product group $G\ltimes H$ if the left action \eqref{Lieact-left-gr} is trivial, and to the (left-handed) semi-direct product group $G\rtimes H$ in case the right action \eqref{Lieact-right-gr} is trivial.

Furthermore, the Lie group analogue of \cite[Prop. 8.3.2]{Majid-book} is given in \cite[Prop. 6.2.15]{Majid-book}, which we also record.

\begin{proposition} \label{universal-prop-gr}
Given a Lie group $K$, with two subgroups $G,H \subseteq K$, if $K\cong G\times H$ as manifolds through 
\[
G\times H\to K, \qquad (x,y)\mapsto xy,
\]
then $(G,H)$ is a matched pair of Lie groups, and moreover $K\cong G\bowtie H$ as Lie groups. The mutual actions, then, are given by 
\begin{equation} \label{mab-defn-gr}
yx=(y\rt x)(y\lt x) \in K,
\end{equation}
for any $x\in G$, and any $y\in H$.
\end{proposition}

\subsection{The double cross product realization}\label{subsect-decomp-diff-can}~

In the present subsection we shall investigate, following \cite[Prop. 2.1]{MoscRang09} in which the case of $\C{Q}=\B{R}^n$ is treated in detail, the double cross product realization of the subgroup ${\rm Diff}_{\rm can,\t}(T^\ast\C{Q}) \subseteq {\rm Diff}_{\rm can}(T^\ast\C{Q})$ of canonical diffeomorphisms that preserve the canonical (Liouville) 1-form. We note from \cite[Prop. 6.3.2]{MarsdenRatiu-book} that these precisely are the ones preserving the cotangent fibers.

Along the lines of \cite[Sect. IV.12]{KoMiSl93}, let $J^\infty_{(q,0)}(T^\ast\C{Q})$ denote the set of all infinite jets $J^\infty_{(q,0)}\Phi$ of diffeomorphisms $\Phi\in {\rm Diff}_{\rm can}(T^\ast\C{Q})$, of partial derivatives with respect to the momentum variables, at $(q,0) \in T^\ast\C{Q}$. Accordingly, letting
\[
S := \{\vp\in {\rm Diff}_{\rm can,\t}(T^\ast\C{Q}) \mid J^\infty_{(q,0)}(\vp) = J^1_{(q,0)}(\vp),\quad \forall\,q\in \C{Q}\},
\]
and
\[
N := \{\psi\in {\rm Diff}_{\rm can,\t}(T^\ast\C{Q}) \mid J^1_{(q,0)}(\psi) = J^1_{(q,0)}(\Id),\quad \forall\,q\in \C{Q}\},
\]
which are both clearly subgroups, we achieve 
\[
S\times N\cong {\rm Diff}_{\rm can,\t}(T^\ast\C{Q}), \qquad (\vp,\psi)\mapsto \vp\circ\psi,
\] 
the proof of which is verbatim to that of \cite[Prop. 2.1]{MoscRang09}, see also \cite{Kac68}. As such, Proposition \ref{universal-prop-gr} yields a Kac-type decomposition
\begin{equation}\label{Kac-decomp-Diff}
{\rm Diff}_{\rm can,\t}(T^\ast\C{Q}) \cong S\bowtie N.
\end{equation}

Furthermore, for any $0\leq s\leq r$ and any $\Phi\in {\rm Diff}_{\rm can,\t}(T^\ast\C{Q})$, the projections $\pi^r_s:J^r_{(q,0)}(\Phi)\mapsto J^s_{(q,0)}(\Phi)$ of $r$-jets into $s$-jets, see for instance \cite[Subsect. 12.2]{KoMiSl93}, endow ${\rm Diff}_{\rm can,\t}(T^\ast\C{Q})$ with the structure of an inverse limit of Lie groups, for the details of which we refer the reader to \cite{Ster61}.  

Accordingly, $S \subseteq {\rm Diff}_{\rm can,\t}(T^\ast\C{Q})$ has itself the structure of a semi-direct product, which has already appeared in \cite{MarsRatiWein84}, in the study of the compressible fluid motion.

On the Lie algebra level, $\C{F}_0(T^\ast\C{Q})$ admits, by construction a natural filtration (see, for instance, \cite{FuksGelfKali72,Perchik-PhD-thesis,Perc76}) based on the degrees of the momentum variables. More precisely,
\[
\C{F}_0(T^\ast\C{Q}) =  L_{-1} \supseteq L_0 \supseteq \ldots \supseteq L_j \supseteq \ldots
\]
where $L_j$ is the set of functions that vanish on $\C{Q}\times \{0\}$ to the order $j$ (of partial differentiation with respect to the momentum variables). The Poisson bracket, then, satisfies 
\[
\{L_r,L_s\} \subseteq L_{r+s},
\]
that is, $\C{F}_0(T^\ast\C{Q})$ becomes a filtered Lie algebra, and hence, an inverse limit of Lie algebras through the projections $\pi_s^r:L_r\to L_s$ for any $0\leq s\leq r$.

Now, it follows from the exponantiation of the Lie algebras of vector fields to the Lie groups of diffeomorphisms, see for instance \cite{MoscRang09} or \cite[Prop. 4]{esen2012geometry}, that the filtrations introduced above are preserved. As such, $\G{s}^c$ exponentiates into $S$, and $\G{n}^c$ into $N$. In other words, the Lie algebra of ${\rm Diff}_{\rm can,\t}(T^\ast\C{Q}) \cong S\bowtie N$ may be identified with $\G{X}_{\rm Ham,0}(T^\ast\C{Q}) \cong \G{s}^c\bowtie \G{n}^c$.

On the other hand, since (infinite) jets are defined as the quotients of the germs of diffeomorphisms (by those whose partial derivatives vanish at all orders), the Lie algebra counterpart of \eqref{Kac-decomp-Diff} may be seen best through
\begin{equation}
\begin{split}
\G{s}^\infty = &\{h \in \C{F}(T^\ast\C{Q}) \mid  \frac{\p^k h}{\p p_{i_1}\ldots \p p_{i_k}}(q,0)=0,\,\, k\geq 2\},
\\
\G{n}^\infty = &\{h \in \C{F}(T^\ast\C{Q}) \mid h(q,0)=0\,,\, \frac{\p h}{\p p_i}(q,0)=0\}
\end{split}
\end{equation}
so that $\G{s}^\infty+\G{n}^\infty =\C{F}(T^\ast\C{Q})$, and that  $\G{s}^\infty\cap\G{n}^\infty = m_{\C{Q}\times \{0\}}^\infty$. As such, 
\begin{equation} \label{decoco}
\C{F}_0(T^\ast\C{Q}) \cong \C{F}(T^\ast\C{Q}) / m_{\C{Q}\times \{0\}}^\infty \cong \G{s}^\infty / m_{\C{Q}\times \{0\}}^\infty \oplus \G{n}^\infty / m_{\C{Q}\times \{0\}}^\infty.
\end{equation}

We conclude with the following remarks.

\begin{remark}
A quick comparison of \eqref{Coad-VF-Ham} and \eqref{coad-Ham-decomposed}, or equivalently \eqref{Coad-VF-} and \eqref{Coad-VF-II}, reveals the non-trivial effect of the graded structure, of $\C{F}_0(T^\ast\C{Q})$ and $\G{X}_{\rm Ham,0}(T^\ast\C{Q})$ respectively, in explicit calculations. More importantly, the inverse limit structures of ${\rm Diff}_{\rm can,\t}(T^\ast\C{Q})$ and $\C{F}_0(T^\ast\C{Q})$, provides a promising avenue for the symplectic and Poisson reductions via finite dimensional Lie groups (in this infinite dimensional setting).
\end{remark}

\begin{remark}
The double cross product decomposition \eqref{Kac-decomp-Diff} above can also be repeated for the other diffeomorphism groups such as the group of diffeomorphisms preserving a volume form, or those preserving the contact form. These cases are both studied in \cite{MoscRang09} in detail, for $\C{Q}=\B{R}^n$ and $\C{Q}=\B{R}^{2n+1}$, respectively.
\end{remark}

\section{Conclusions and Discussion}

We have stated and proved novel results such as the (matched pair) decomposition of the Vlasov equation, along with the dynamics of its kinetic moments. More precisely in Proposition \ref{mpdTQ} we proposed a matched pair Lie algebra decomposition of the symmetric contravariant vector fields. In Proposition \ref{prop-dualaction}, we computed the induced dual actions, while in Proposition \ref{cross-act} we obtained the induced cross actions. We exhibit the matched pair decomposition of the dynamics of kinetic moments in \eqref{coad-mp-kinetic-moments}. In  Proposition \ref{prop-SC-Poisson-identify}, a Lie algebra homomorphism from the symmetric contravariant tensor fields to the algebra $\C{F}(T^\ast\C{Q})$ of functions on the cotangent bundle is introduced, via which the matched pair decomposition of the subalgebra $\C{F}_0(T^\ast\C{Q}) \subseteq \C{F}(T^\ast\C{Q})$ was obtained in \eqref{C-decomp-1-}. Dually, we have presented the matched pair decomposition of the Vlasov equation in Corollary \ref{coroll-matched-Vlasov}. In order to transfer all these discussions to the level of Hamiltonian vector fields and the momentum-Vlasov equations, a Lie algebra homomorphism, \text{{\small GCCL}}, has been introduced in \eqref{GCCL}. Accordingly, the matched pair decomposition of the Hamiltonian vector fields has been derived in \eqref{Ham-vf} and \eqref{formal-Ham-vf}.  Thus, we could realize the momentum-Vlasov equations \eqref{momvla} as a matched pair Lie-Poisson system in Corollary \ref{coroll-matched-mom-Vlasov}. We illustrate the relation with the matched pair Lie algebras via the commutative diagram 
\begin{equation*} 
\xymatrix{
{\G{T}\C{Q}} \ar[d]^{\kappa} \ar@/_4pc/[dd]_{\text{\small GCCL}} \ar@{=}[rr]^{\eqref{GTQ-matched-pair}} && \G{s}\bowtie \G{n}  \ar[d]^{\kappa} &
\\
\C{F}_0(T^*\C{Q}) \ar[d]^\vp\ar@{=}[rr]^{\eqref{C-decomp-1-} } && \hat{\G{s}} \bowtie \hat{\G{n}} \ar[d]^\vp
&\\
{\G{X}}_{\mathrm{Ham},0}(T^*\C{Q})\ar@{=}[rr]^{\eqref{formal-Ham-vf}} && \G{s}^c\bowtie \G{n}^c \ar@{=}[r] &\hat{\G{s}}/\mathbb{R}\bowtie \hat{\G{n}},
}
\end{equation*}
while for the relation with the matched pair Lie-Poisson spaces we record 
\begin{equation*} 
\xymatrix{
{\G{T}^*\C{Q}}  \ar@{=}[rr]^{\eqref{GTstarQ-matched-pair}} && \G{s}^*\oplus \G{n}^*   &
\\
\C{F}_0^*(T^*\C{Q})  \ar[u]_{\kappa^*} \ar@{=}[rr]^{\eqref{F_0-dual-decomp}} && \hat{\G{s}}^* \oplus \hat{\G{n}}^*    \ar[u]_{\kappa^*} &
\\
{\G{X}}_{\mathrm{Ham},0}^*(T^*\C{Q})  \ar[u]_{\vp^\ast}
\ar@/^4pc/[uu]^{\text{\small GCCL}^*}
\ar@{=}[rr]^{\eqref{X_Ham-dual-decomp}} && (\hat{\G{s}}/\mathbb{R})^* \oplus (\hat{\G{n}}^c )^* \ar[u]_{\vp^\ast} \ar@{=}[r] & (\G{s}^c)^* \oplus (\G{n}^c )^*.
}
\end{equation*}
We, finally, include below an incomplete list of future works, along with open problems, related to the content of the present paper. 

\subsubsection*{\textbf{Cocycle double cross sum Lie algebras. 10-moment approximations. BBGKY hierarchy.}} Let us note that ${\mathfrak{T}\mathcal{Q}}$ can be written as a (vector space) direct sum
\begin{equation} \label{decomp-TQ-a}
{\mathfrak{T}\mathcal{Q}}=\sum_{k=0}^a \, \mathfrak{T}^{k}\mathcal{Q}\oplus \sum_{k=a+1}^\infty\,\mathfrak{T}^{k}\mathcal{Q},
\end{equation}
  for all integers $a$. In Proposition \ref{mpdTQ}, it is established that, the direct sum \eqref{decomp-TQ-a} turns out to be a matched pair for $a=1$. Similarly, one can show that the direct sum \eqref{decomp-TQ-a} is a matched pair for $a=0$ as well, though, in the latter case the Lie bracket on the first constitutive subalgebra $\mathfrak{T}^0\mathcal{Q}=\C{F}(\mathcal{Q})$ is trivial. 
  
  On the other hand, for $a\geq 2$ the decomposition \eqref{decomp-TQ-a} fails to be a matched pair Lie algebra decomposition as the subspace $\sum_{k=0}^a \, \mathfrak{T}^{k}\mathcal{Q}$ is no longer a Lie subalgebra of ${\mathfrak{T}\mathcal{Q}}$. However, it was observed in \cite{AgorMili14} that these ``extended'' decompositions for $a\geq 2$ are also Lie algebras, that we call ``cocycle double cross sum'' Lie algebras in \cite{EsGuSu20}. Such extended structures provide a unifying generalization of both the matched pair Lie algebras, and the 2-cocycle extensions of Lie algebras. What we also observe in \cite{EsGuSu20} is that among other examples of cocycle double cross sum Lie algebras are the enveloping algebras of Lie-Yamaguti algebras \cite{Kikk75,KinyWein01,Yama57}, and that the universal enveloping algebra of a cocycle double cross sum Lie algebra is a Brzeziński crossed product (with a coalgebra), \cite{Brze97-II}.
  
The algebraic/geometric analysis of the decomposition of the kinetic moments, that corresponds to degrees $a \geq 2$ in \eqref{decomp-TQ-a}, seems also a promising research area in fluid and plasma theories. For instance, in the case $a=2$, one has $10$-moment kinetic theory \cite{EsGrGuPa19} which paves the way towards to whole Grad hierarchy \cite{grad1965boltzmann} including the
entropic moments \cite{grmela2017hamiltonian}. 
We refer to \cite{Le96,PeChMoTa15,Pe90} for a collection of the related works on the kinetic moments. 

The cocycle double cross sum Lie algebra construction of \cite{EsGuSu20} may have more applications in plasma physics. A Hamiltonian analysis of the well-known BBGKY (Bogoliubov-Born-Green-Kirkwood-Yvon) hierarchy of the plasma dynamics \cite{Cercignani1997} has been achieved in \cite{MaMoWe83}, wherein the dynamics of the hierarchy is written as a Lie-Poisson equation. This Lie-Poisson framework does not admit a matched pair decomposition, hence it lies outside of the scope of the present paper. Nevertheless, it fits into the cocycle double cross sum decomposition we propose in \cite{EsGuSu20}. 

\subsubsection*{\textbf{Decomposition of the Euler-Poincar\'{e} realization of the geodesic Vlasov equation.} } A pure quadratic Lagrangian functional on the space $\mathfrak{X}_{\mathrm{\mathrm{ham}}}(T^\ast \mathcal{Q})$ of Hamiltonian vector fields was introduced in \cite{holm2009geodesic}. Then referring to this metric, the geodesic Vlasov equation has been studied in the Euler-Poincar\'{e} formalism. On the other hand, we proposed the abstract theory of matched Euler-Poincar\'{e} equations in \cite{EsenSutl17}. Using the matched pair decomposition of the Hamiltonian vector fields in \eqref{Ham-vf} and \eqref{formal-Ham-vf}, it is possible to apply the theory in \cite{EsenSutl17} to the geodesic Vlasov equation. Moreover, it is then natural to connect this Lagrangian picture to  the Hamiltonian one by a proper and non-degenerate (thanks to the quadraticity of the Lagrangian) Legendre transformation. A similar, but relatively harder problem (one with the regularity of the Lagrangian is not assumed) is posed in the following paragraph.  
  
\subsubsection*{\textbf{The (inverse) Legendre transformation of the Vlasov plasma.}} Finding a Legendre transformation between the Euler-Poincaré and the Lie-Poisson formulations of the Poisson-Vlasov equations cannot be achieved in a straightforward manner, as a result of the degeneracy of the Hamiltonian function(al). One way to overcome this difficulty is to use the Tulczyjew's triplet which allows the Legendre transformation for singular systems as well, \cite{tulczyjew1977legendre}. For the particular form of the Tulczyjew's triplet for Lie groups, which is initiated by the same motivational question, we refer the reader to \cite{esen2014tulczyjew,esen2017tulczyjew}, see also \cite{grabowska2016tulczyjew}.  
The matched pair strategy may, on the other hand, be used to construct a proper Tulczyjew's triplet for the Vlasov plasma.

\section{Acknowledgment}

The first named author (OE) is grateful to Prof. Hasan Gümral for enlightening discussions on the Vlasov plasma especially for momentum-Vlasov dynamics.  OE is also grateful to Prof. Miroslav Grmela, Prof. Michal Pavelka, Prof. Petr Vágner for enlightening discussions on (ir)reversible plasma dynamics. Both authors are grateful Prof. Mansur Ismailov for discussions done on some functional analytical details. 
Both authors gratefully acknowledge the support by T\"UB\.ITAK (the Scientific and Technological Research Council of Turkey) under the project "Matched pairs of Lagrangian and Hamiltonian Systems" with the project number 117F426.

\bibliographystyle{amsplain}
\bibliography{references}

\end{document}